\documentclass[11pt]{article}

\ifdefined\CR
\usepackage{ltexpprt}
\newenvironment{definition}{\begin{Definition}}{\end{Definition}}
\newtheorem{claim}{Claim}[section]
\newenvironment{proofof}[1]{\rm \trivlist \item[\hskip \labelsep{\it #1.\/}]}{\outerparskip 0pt\endtrivlist}
\else
\usepackage{amsthm}
\usepackage{fullpage}
\usepackage{typearea}
\typearea{15}
\newtheorem{theorem}{Theorem} 
\newtheorem{lemma}{Lemma}[section]

\newtheorem{definition}[lemma]{Definition}
\newtheorem{claim}[lemma]{Claim}
\newenvironment{proofof}[1]{\begin{proof}[#1]}{\end{proof}}
\fi

\usepackage{amssymb, amsmath}
\usepackage{graphicx}
\usepackage{xspace}
\usepackage{color}
\usepackage{xcolor}
\usepackage[pagebackref,letterpaper=true,colorlinks=true,pdfpagemode=none,urlcolor=blue,linkcolor=blue,citecolor=violet,pdfstartview=FitH]{hyperref}
\usepackage{enumitem}
\usepackage{thm-restate}
\usepackage{mathrsfs}

\newcounter{prop}
\newenvironment{properties}
{
\addtocounter{prop}{1}
\begin{enumerate}[labelindent=0pt,label=(\Alph{prop}\arabic*),itemindent=1em,itemsep=-1pt]
}
{
\end{enumerate}
}

\newcommand{\set}[1]{\left\{#1\right\}}
\newcommand{\cardinal}[1]{\left|#1\right|}

\DeclareMathOperator{\E}{\mathbb E}

\newcommand{\R}{{\mathbb R}}

\newcommand{\calA}{{\mathscr A}}
\newcommand{\calB}{{\mathscr B}}

\newcommand{\cA}{{\mathscr A}}
\newcommand{\cB}{\mathscr{B}}
\newcommand{\cC}{{\mathscr C}}
\newcommand{\cD}{\mathscr{D}}
\newcommand{\cE}{{\mathscr E}}

\newcommand{\cI}{{\mathsf I}}

\newcommand{\sfH}{{\mathsf H}}
\newcommand{\sfL}{{\mathsf L}}

\newcommand{\eps}{\varepsilon}
\renewcommand{\epsilon}{\varepsilon}

\newcommand{\poly}{\mathrm{poly}}
\newcommand{\polylog}{\mathrm{polylog}}

\newcommand{\x}{{\mathbf x}}

\newcommand{\Exp}{\EX}
\newcommand{\EX}{\hbox{\bf E}}

\newcommand{\OPT}{{\mathsf{OPT}}}
\newcommand{\LP}{{\mathsf{LP}}}

\newcommand{\floor}[1]{\left\lfloor#1\right\rfloor}

\newcommand{\ignore}[1]{}
\renewcommand{\epsilon}{\varepsilon}

\newcommand{\Sec}[1]{\texorpdfstring{\hyperref[sec:#1]{\S\ref*{sec:#1}}}{\S\ref*{sec:#1}}} 
\newcommand{\Thm}[1]{\texorpdfstring{\hyperref[thm:#1]{Theorem~\ref*{thm:#1}}}{Theorem~\ref*{thm:#1}}} 
\newcommand{\Lem}[1]{\texorpdfstring{\hyperref[lem:#1]{Lemma~\ref*{lem:#1}}}{Lemma~\ref*{lem:#1}}} 
\newcommand{\Clm}[1]{\texorpdfstring{\hyperref[clm:#1]{Claim~\ref*{clm:#1}}}{Claim~\ref*{clm:#1}}} 
\newcommand{\App}[1]{\hyperref[app:#1]{Appendix~\ref*{app:#1}}} 

\allowdisplaybreaks

\title{On $(1,\epsilon)$-Restricted Assignment Makespan Minimization}
\author{Deeparnab Chakrabarty\footnote{Microsoft Research. {\tt dechakr@microsoft.com}} \and Sanjeev Khanna\thanks{Department of Computer and Information Science, University of Pennsylvania,
Philadelphia, PA 19104. Email: {\tt sanjeev@cis.upenn.edu}.  Supported in part by National Science Foundation grant CCF-1116961.}  \and Shi Li\footnote{TTIC. {\tt shili@ttic.edu}}}

\begin{document}
\date{}
\maketitle

\begin{abstract}
Makespan minimization on unrelated machines is a classic problem in approximation algorithms. No polynomial time $(2-\delta)$-approximation algorithm is known for the problem for constant $\delta> 0$. This is true even for certain special cases, most notably the {\em restricted assignment} problem where each job has the same load on any machine but can be assigned to one from a specified subset. Recently in a breakthrough result, Svensson~\cite{Sve11} proved that the integrality gap of a certain configuration LP relaxation is upper bounded by $1.95$ for the restricted assignment problem; however, the rounding algorithm is {\em not known} to run in polynomial time.\smallskip

In this paper we consider the $(1,\epsilon)$-restricted assignment problem where each job is either heavy ($p_j = 1$) or light ($p_j = \eps$), for some parameter $\eps > 0$. 
Our main result is a $(2-\delta)$-approximate \emph{polynomial time} algorithm for the $(1,\epsilon)$-restricted assignment problem for a fixed constant $\delta> 0$. Even for this special case, the best polynomial-time approximation factor known so far is $2$. We obtain this result by rounding the configuration LP relaxation for this problem.
A simple reduction from vertex cover shows that this special case remains NP-hard to approximate to within a factor better than $7/6$.
\end{abstract}

\setcounter{page}{0}
\thispagestyle{empty}
\ifdefined\CR

\else
\newpage
\fi

\section{Introduction}

In the makespan minimization problem, we are given a set $M$ of $m$ machines, and a set $J$ of $n$ jobs where a job $j$ contributes a load of $p_{ij}$ to a machine $i$, if assigned to it. The goal is to assign each job to a machine, so that the maximum load on any machine is minimized. Formally, one seeks an allocation $\sigma:J\to M$ minimizing $\max_{i\in M} \sum_{j:\sigma(j)=i} p_{ij}$. In 1990, Lenstra, Shmoys, and Tardos \cite{LST90} gave a $2$-approximation for the problem, and showed that it is NP-hard to obtain an approximation factor better than $3/2$. 
Closing this gap is an outstanding problem in approximation algorithms.

In order to understand the problem better, researchers have focused on special cases. 
The most notable among them is the {\em restricted assignment} problem. 
In this problem, each job $j\in J$ has an inherent load $p_j$ but it can be assigned only to a machine  from a specified subset. Equivalently, each $p_{ij} \in \{p_j,\infty\}$ with $p_{ij} = \infty$ for machines $i$ which $j$ cannot be assigned to.
 The hardness of $3/2$ carries over to the restricted assignment problem and no  polynomial time $(2-\delta)$-algorithm is known for any constant $\delta > 0$.
In a breakthrough, Svensson~\cite{Sve11} proved that the integrality gap of a certain configuration LP for the restricted assignment problem is at most  $33/17$.  Svensson's result can thus be used to efficiently {\em estimate the value} of the optimum makespan to within a factor of $33/17$; however, no {\em polynomial time} algorithm to compute a corresponding schedule is known. Nonetheless, it gives hope\footnote{Without discussing this technically, we refer the reader to the articles by Feige~\cite{F08}, and Feige and Jozeph~\cite{FJ14}.} that the restricted assignment case may have a polynomial time `better-than-factor-$2$' algorithm. 

Our paper makes progress on this front.  We study the $(1,\epsilon)$-restricted assignment problem, in which 
all jobs fall in two classes: heavy or light. Every heavy job has load, $p_j = 1$, and each light job has $p_j = \eps$, for some parameter $\eps > 0$, and the goal is to find a schedule which minimizes the makespan.
We give a $(2-\delta^*)$-approximate \emph{polynomial time} algorithm for the $(1,\epsilon)$-restricted assignment problem for a constant $\delta^*\!>\! 0$.

The $(1,\epsilon)$-case is interesting because in some sense it is the simplest case which we do not understand. If all jobs have the same size, the problem becomes a matching problem. If there are two job sizes, we can assume they are $1$ and $\epsilon < 1$ by scaling. The reader should think of $\eps$ as a number that tends to $0$, as there is a simple $(2-\epsilon)$-approximation if each job has size either $1$ or $\eps$ (see \App{large-eps}).  The $(1,\epsilon)$-restricted assignment problem is already hard -- it is NP hard to obtain an approximation factor $<7/6$ for this problem (see \App{hardness}), and no $(2-\delta)$-approximation is known for any $\delta$ independent of $\eps$.  It is plausible that an understanding of the $(1,\epsilon)$-case might lead to an understanding of  the restricted assignment case; indeed,
 Svensson \cite{Sve11}, in his paper, first gives an improved integrality gap of $(5/3+\eps)$ for this special case before proving  his result for the general case. 

\begin{theorem}(Main Result.)
\label{thm:main}
 There is a polynomial time $(2-\delta^*)$-approximation algorithm, where $\delta^*>0$ is a specific constant,  for makespan minimization in the $(1,\epsilon)$-restricted assignment case.
\end{theorem}

\subsection{Our Techniques}
For concreteness, let us assume for now that the optimal makespan is $1$.
Note that once we have an assignment of the heavy jobs to machines, we can decide in polynomial time whether there is a (fractional) allocation of light jobs so that the total makespan is at most $(2-\delta)$ or not, for any $\delta>0$. Such an assignment of heavy jobs is called a {\em $\delta$-good} assignment, and given such an assignment one can recover an integral assignment for light jobs as well such that the total load on any machine is at most $(2 - \delta + \eps)$. 
The rounding process to recover a $\delta$-good assignment proceeds in three phases. 

In the first phase, we `reduce' our instance to a {\em canonical instance} where for each heavy job there is a distinct (private) set of machines to which it can be assigned to, and for each unassigned light job, there are at most  two machines to which it can be assigned to. Such a  pre-processing technique has also been used in 
tackling the max-min allocation problem~\cite{BS06,CCK09}. There are two main parameters of interest in a canonical instance, namely, a parameter $p$ that asserts that the positive fractional assignment of a heavy job to a machine is at least $1/p$, and a parameter $q$, that asserts that the total load of light jobs shared by any two machines is either $0$ or is at least $1/q$.

The second phase of the rounding process is a {\em coarsening} of the parameters $p$ and $q$ of the canonical instance where we ensure that whenever a heavy job is fractionally assigned to a machine, the assignment is sufficiently large (at least $1/q_0$ for some constant $q_0$). Furthermore, the total light load shared by any pair machines is either $0$ or at least  $1/q_0$. The flip side is the total fractional load on a single machine could increase from $1$ to roughly $1 + \sqrt{\log q_0/q_0}$.
The technique used to achieve this coarsening is essentially due to Feige~\cite{Fei08} in his work on max-min allocation, and is done by iteratively assigning heavy jobs and `merging' light jobs. The proof uses the asymmetric Lovasz Local Lemma (LLL), and 
polynomial time algorithms to do the same are guaranteed by the recent works of Moser and Tardos~\cite{MT10} and Haeupler, Saha, and Srinivasan~\cite{HSS11}.

The heart of our approach and our main technical contribution is the third and the final phase of rounding. At this point we have a canonical instance where each heavy job is assigned to a constant number $q_0$ of machines that are private to it, each light job has constant size, and is shared between at most  two machines. Note that if the fractional load on each machine was at most $1$, then things are trivial -- assign the heavy job to the machine which is fractionally assigned more than $1/q_0$ of it, and the total load on it become $2-1/q_0$. 
However, the second step has increased the fractional load from $1$ to $1\!\!+\!\!\sqrt{\log q_0/q_0}$, and this `extra' load swamps the gain of $1/q_0$.
This issue does not arise in the max-min allocation where one targets a constant factor; however, it defeats our goal of beating the 
factor $2$-approximation for makespan minimization.

Nevertheless, if we could find an assignment such that the total light load on any machine receiving a heavy job is at most $1\!-\!1/{\small \polylog q_0}$, then we are in good shape, and this is what we do.
We find such an assignment by randomized rounding and again use the (asymmetric) LLL. A key and difficult aspect of making this entire approach work is to have only a small degree of dependence between various ``bad'' events in the final rounding step. This reduction in dependence is the essence of our approach, and is accomplished by the structure of canonical instances, 
and further simplifying this structure before picking the final random assignment.

\subsection{Relevant Related Work}
We briefly note some other works on makespan minimization. Ebenlendr, Krc\'{a}l, and Sgall~\cite{EKS08} consider the special case of the restricted assignment makespan minimization problem where each job could be assigned to at most $2$ machines, and design a polynomial time $1.75$-approximation algorithm for the same. Interestingly, even when jobs can go to at most two machines, the general makespan minimization problem seems difficult; Verschae and Wiese~\cite{VW11} show that the configurational LP has integrality gap tending to $2$. Kolliopoulos and Moysoglou~\cite{KM13} consider the restricted assignment problem with two jobs sizes as well; they show that Svensson's estimation algorithm can be improved to $1.883$ for this case. See \App{large-eps} for a slightly better factor.

The `dual' problem to makespan minimization, max-min allocation, where jobs need to be allocated to maximize the minimum load, has seen interesting developments recently. A constant factor approximation is known for the `restricted assignment' case (the so-called {\em Santa Claus} problem), where each job has the same load for all the machines it can go to. This follows from the works of Bansal and Sviridenko~\cite{BS06}, Feige~\cite{Fei08}, and the constructive LLL version due to~\cite{MT10,HSS11}. Our work closely follows this line and exhibits its utility for the makespan minimization problem. Another line of work on the Santa Claus problem is via local search; Asadpour, Feige, and Saberi~\cite{AFS08} show an upper bound of $4$ on the integrality gap via a not-known-to-be-polynomial time local search algorithm. Polacek and Svensson~\cite{PS12} use these ideas to give a {\em quasipolynomial} time, $4+\epsilon$-approximation. 
Very recently, Annamalai, Kalaitzis, and Svensson~\cite{AKS15} improve this to get a polynomial time $13$-approximation.
For the {\em general} max-min allocation problem, Chakrabarty, Chuzhoy, and Khanna~\cite{CCK09}, improving upon earlier results by Asadpour and Saberi~\cite{AS07} and Bateni, Charikar, and Guruswami~\cite{BCG09}, give a $O(n^\epsilon)$-approximation algorithm which runs in $O(n^{1/\epsilon})$-time.

\section{Linear Programming Relaxation}
\label{sec:LP}

Recall that we denote the set of all machines by $M$ and the set of all jobs by $J$. We assume that $J$ is partitioned into the set of heavy jobs by $J_\sfH$, and the set of light jobs by $J_\sfL$.  We consistently use $j$ to index jobs, and $i,h$ and $k$ to index machines. For any job $j \in J$, we denote by $M_j$ the set of machines to which job $j$ can be assigned.

Given a guess $T\geq 1$ for the optimal makespan, the {\em configuration LP} w.r.t. $T$ is as follows. For every machine $i$, $\cC_i$ contains subsets of jobs of total load at most $T$ which can be assigned to $i$ . We have a variable $y_{i,C}$ for each machine $i$ and subset $C\in \cC_i$. 
\begin{alignat}{2}
\textstyle \sum_{C\in \cC_i} y_{i,C} & = 1& \qquad \forall i\in M \label{eq:ConfLP1}\tag{Conf LP 1}\\
\textstyle \sum_{i\in M}\sum_{C\in \cC_i:j\in C} y_{i,C} & = 1& \qquad \forall j\in J \label{eq:ConfLP2} \tag{Conf LP 2}
\end{alignat}
\noindent
Given an instance $\cI$, we let $\OPT_f$ be the {\em smallest} $T$ for which the configuration LP has a feasible solution; $\OPT_f$ can be found by binary search.
Trivially, $1\leq \OPT_f \leq \OPT$, where $\OPT$ denotes the optimal (integral) makespan. 

In this paper, we use the following simpler parametrized LP relaxation $\LP(\rho,\delta)$ tailored for $(1,\eps)$-instances. 
\ifdefined\CR
\begin{alignat}{2}
\textstyle \sum_{i \in M_j}x_{i,j} &= 1 & \qquad \forall j &\in J  \label{CLP:job-covered}\\
\textstyle \sum_{j \in J_\sfH : i \in M_j}x_{i,j} &= z_i & \qquad \forall i &\in M \label{CLP:machine-big-job}\\
\textstyle z_i & \leq 1 & \qquad \forall i&\in M  \label{CLP:zi-atmost-1} \\
\textstyle z_i + \epsilon \sum_{j \in J_\sfL : i \in M_j} x_{i,j} &\leq 1 + \rho\delta &\qquad \forall i &\in M \label{CLP:machine-small-job}\\
(1-\rho) z_i + x_{i,j} &\leq 1 & \qquad \forall j &\in J_\sfL, i \in M_j \label{CLP:compact}\\
x_{i,j}, z_i &\geq 0 & \qquad \forall j &\in J, i \in M_j \nonumber
\end{alignat}
\else
\begin{minipage}{0.37\textwidth}
\begin{alignat}{2}
\textstyle \sum_{i \in M_j}x_{i,j} &= 1 & \qquad \forall j &\in J  \label{CLP:job-covered}\\
\textstyle \sum_{j \in J_\sfH : i \in M_j}x_{i,j} &= z_i & \qquad \forall i &\in M \label{CLP:machine-big-job}\\
\textstyle z_i & \leq 1 & \qquad \forall i&\in M \label{CLP:zi-atmost-1}
\end{alignat}
\end{minipage}
\begin{minipage}{0.58\textwidth}
\begin{alignat}{2}
\textstyle z_i + \epsilon \sum_{j \in J_\sfL : i \in M_j} x_{i,j} &\leq 1 + \rho\delta &\qquad \forall i &\in M \label{CLP:machine-small-job}\\
(1-\rho) z_i + x_{i,j} &\leq 1 & \qquad \forall j &\in J_\sfL, i \in M_j \label{CLP:compact}\\
x_{i,j}, z_i &\geq 0 & \qquad \forall j &\in J, i \in M_j \nonumber
\end{alignat}
\end{minipage}
\fi

\medskip
To get some intuition, 
consider the LP with $\rho=\delta =  0$.
We claim there exists a  feasible solution if $\OPT = \OPT_f = 1$. In this case, it must be that every machine either gets one heavy job or 
at most $\floor{1/\eps}$ light jobs. In particular, any machine getting a light job {\em cannot} get a heavy job.
Constraint~\eqref{CLP:compact} encodes this. 
The connection between $\LP(\rho,\delta)$ and the configuration LP is captured by the following lemma.
\begin{restatable}{lemma}{clp}
\label{lemma:reduce-to-feasible-instances}
\label{lem:reduce-to-feasible-instances}
Given an $(1, \eps)$-restricted assignment instance $\cI$ 
with $\OPT_f\leq 1+\rho\delta$, there is a polynomial time procedure which 
returns another $(1,\eps)$-instance $\cI'$ which has a feasible solution to $\LP(\rho,\delta)$.
Furthermore, given a schedule for $\cI'$ with makespan $T$, the procedure returns a schedule for $\cI$ of makespan $\leq T+\delta$.
\end{restatable}

\begin{proof}
Let $y$ be the solution to the configuration LP at $\OPT_f \leq 1+\rho\delta$.
Call a configuration $C$ {\em heavy} if it contains a heavy job. 
Define
$z_i := \sum_{C\text{ is heavy}}y_{i,C}$ for all $i$, and $x_{i,j} := \sum_{C:j\in C} y_{i,C}$ for all $i,j$.
Note that  for all $i \in M$, we have 
$z_i + \eps\sum_{j\in J_\sfL} x_{i,j} \leq 1+\rho\delta$ since each configuration has load $\leq 1+\rho\delta$.

Now, if for some light job $j$, $\sum_{C \textrm{ heavy}:j\in C} y_{i,C}   > \rho z_i$, we remove $j$ from $J_\sfL$ and set $\sigma(j) = i$. 
Let $J'_\sfL$ be the set of remaining jobs. The new instance is $\cI' = (M,J_\sfH \cup J'_\sfL)$.
For any job $j$ in $J'_\sfL$, 
$x_{i,j} \leq  \sum_{C: \textrm{not heavy}} y_{i,C}  + \rho z_i  =   1 - (1-\rho)z_i$, that is, $(1-\rho)z_i + x_{i,j} \leq 1.$
Thus,  $(z,x)$ is a feasible solution for $\LP(\rho,\delta)$ for $\cI'$. 
Now, given an assignment of jobs for $\cI'$, we augment it to get one for $\cI$ by assigning job $j\in J_\sfL\setminus J'_\sfL$ to $\sigma(j)$.
Note
\ifdefined\CR
\begin{align*}
 &\quad \eps\cardinal{\sigma^{-1}(i)} < \eps\sum_{j\text{ light}} \sum_{C \textrm{ heavy }: j\in C}  \frac{y_{i,C}}{\rho z_i}\\
 &= \frac{\eps}{\rho} \sum_{C: \textrm{ heavy}} \frac{y_{i,C}}{z_i} \cardinal{C \cap J_\sfL} \leq \delta
\end{align*}
\else
\begin{alignat}{1}
 \eps\cardinal{\sigma^{-1}(i)} < \eps\sum_{j\text{ light}} \sum_{C \textrm{ heavy }: j\in C}  \frac{y_{i,C}}{\rho z_i} = \frac{\eps}{\rho} \sum_{C: \textrm{ heavy}} \frac{y_{i,C}}{z_i} \cardinal{C \cap J_\sfL} \leq \delta\nonumber
\end{alignat}
\fi
since $\eps\cardinal{C \cap J_\sfL}\leq \rho\delta$ for heavy $C$.
\end{proof}

The remainder of the paper is devoted to proving the following theorem.

\begin{theorem} 
\label{thm:main-2}
There is a constant $\delta_0 \in(0,1)$, such that 
given an instance $\cI$ and a feasible solution to $\LP(\rho\!=\!0.6,\delta_0)$, in polynomial time one can obtain a schedule for $\cI$ of makespan at most $(2\!-\!2\delta_0)$.
\end{theorem}

We conclude this section by showing that the preceding theorem suffices to establish our main result.

\begin{proofof}{\bf Proof of Theorem \ref{thm:main}}
Set $\delta^* = \delta_0/2$, where $\delta_0$ is the constant specified in \Thm{main-2}.
Fix an instance $\cI$ and the corresponding $\OPT_f$.
If $\OPT_f > 1 + \rho\delta_0$, then the classic result of Lenstra et al.~\cite{LST90} returns a schedule whose makespan is at most $\OPT_f ~+ 1 \leq \OPT_f\left(1+ \frac{1}{1+\rho\delta_0}\right) \leq \left(2-\delta^*\right)\OPT_f$.
If $\OPT_f \leq 1+ \rho\delta_0$, then \Lem{reduce-to-feasible-instances} can be used to get an instance $\cI'$ for which $\LP(\rho,\delta_0)$ is feasible.
\Thm{main-2} returns a schedule for $\cI'$ with makespan at most $(2 -2\delta_0)$, which when fed to \Lem{reduce-to-feasible-instances} gives
a schedule for $\cI$ of makesan at most $(2 - \delta_0) \leq (2-\delta^*)\OPT_f$ since $\OPT_f\geq 1$.
This proves \Thm{main}.
\end{proofof}

\section{Canonical Instances and \texorpdfstring{$\delta$}{delta}-good Assignments}
\label{sec:canonical_definition}

In this section we introduce the notion of canonical instances and formalize the notion of a $\delta$-good assignment of heavy jobs for these instances.  


In a canonical instance, heavy jobs have size $p_j = 1$. Light jobs can be scheduled fractionally and any light job can be assigned to at most two machines. Thus each light job is of type-$(h, k)$ for some $h, k \in M$; it can only be assigned to $h$ or $k$. It is possible that $h \!=\! k$; when $h \!\neq\! k$, $(h, k)$ and $(k, h)$ are two different job types. Subsequently, it will be clear that these types are differentiated and defined by how the LP assigns the jobs; for now the reader may think of $w_{h,k}$ as the load of light jobs which `belong to $k$ but can be assigned to $h$ if $k$ gets a heavy job'.
  Given $h, k$, we will merge the light jobs of type-$(h, k)$ into a single job of total size equal to the sum of the light jobs. We call this the light load of type-$(h, k)$.
Henceforth, we use ``light load'' instead of ``light jobs'' in a canonical instance. 

\begin{definition}
\label{def:ci-0}
A {\bf canonical instance} is defined by a triple $(\set{M_j: J \in J_\sfH}, w, z)$, where
\begin{properties}
\item for every heavy job $j \in J_\sfH$, $M_j \subseteq M$ is the set of machines that $j$ can be assigned to; for any pair heavy jobs $j \neq j'$, we have $M_j \cap M_{j'} = \emptyset$; \label{property:ci-M}
\item $w \!\in\! \R_{\geq 0}^{M \times M}$ a matrix, where  $w_{h,k}$ is the light load of type-$(h,k)$. If $z_k = 0$ and $h \neq k$, then $w_{h,k} \!= \!0$; if $z_h > 0$ then $w_{h, h} = 0$;  \label{property:ci-w}
\item $z:M \mapsto [0, 0.4]$ is a function on $M$ where $z_i = 0$ if and only if $i \notin \bigcup_{j \in J_\sfH}M_j$. \label{property:ci-z}
 \end{properties}
\end{definition}
\noindent
There is an intrinsic fractional solution defined by the function $z$. If $i \in M_j$ for some $j \in J_\sfH$, then $z_i$ is the fraction of the heavy job $j$ assigned to machine $i$. A heavy job may not be fully assigned, but we will ensure that a decent portion of it is.  If $h\!\neq\! k$, $(1\!-\!z_k)$ fraction the $w_{h,k}$ light load of type-$(h, k)$ is assigned to $k$, and the remaining $z_k$ fraction is assigned to $h$.  The $w_{i,i}$ light load of type-$(i,i)$ is fully assigned to $i$.  
Given a matrix $w \in \R^{M \times M}$, we use the notation $w_{A,B}$ for subsets $A,B\subseteq M$ to denote the sum $\sum_{h\in A,k\in B} w_{h,k}$.   

\begin{definition}
The directed graph $G_w = (M, \{(h,k):h \neq k, w_{h,k} > 0\})$ formed by the support of $w$, with self-loops removed, is called the {\bf light load graph}.
\end{definition}

\begin{definition}
Given a canonical instance $\cI = \left(\set{M_j : j \in {J_\sfH}}, w, z\right)$, we say that $\cI$ is a {\bf $(p, q, \theta)$-canonical instance} for some $p \geq 1, q \geq 1$ and $\theta  \in [0, 0.2)$, if it satisfies the following properties (in addition to Properties~\ref{property:ci-M} to \ref{property:ci-z}):
\begin{properties}
\item for each $i \in M$, either $z_i = 0$, or $z_i \geq 1/p$; 
\label{property:canonical-instance-z-i-large}
\item for every $h, k \in M$, either $w_{h,k} = 0$ or $w_{h, k} \geq 1/q$;
\label{property:canonical-instance-light-load-large}
\item $\sum_{i\in M_j}z_i \geq 0.2 - \theta$, for every $j \in {J_\sfH}$;
\label{property:canonical-instance-big-job-covered}
\item $z_h + \sum_{k\in M}w_{k, h}(1-z_h) + \sum_{k \in M}w_{h,k}z_{k} \leq 1 + \theta$, for every machine $h \in M$.
\label{property:canonical-instance-load-small}
\end{properties}
\end{definition}

Property~\ref{property:canonical-instance-z-i-large} says that none of the heavy job assignments is too small.  Property~\ref{property:canonical-instance-light-load-large} says that any positive load of some type is large. Property~\ref{property:canonical-instance-big-job-covered} says that  a `decent' fraction of each heavy job $j$ is assigned. Property~\ref{property:canonical-instance-load-small} says that the total load assigned to a machine $h \in M$ in the intrinsic fractional solution is bounded.
Our goal is to find a valid assignment $f:J_\sfH \mapsto M$ of heavy jobs to machines which leaves ``enough room'' for the light loads. We say $f$ is {\em valid} if $f(j) \in M_j$ for every $j \in J_\sfH$. This motivates the following definition.

\begin{definition}[$\delta$-good Assignment]
Given a $(p, q, \theta)$-canonical instance and a number $\delta \in (0, 1)$,  a valid assignment $f:J_\sfH\to M$ for a canonical instance is $\delta$-good if all the light loads can be \emph{fractionally} assigned so that each machine has total load at most $2-\delta$.
\end{definition}
\noindent
Define $f({J_\sfH}) = \set{f(j):j\in {J_\sfH}}$. The following theorem (basically Hall's condition)
is a characterization of good assignments. 

\begin{restatable}{theorem}{density}
\label{thm:alpha-good-equivalent-to-no-bad-set}

For a canonical instance, an assignment $f$ of heavy jobs is a $\delta$-good assignment if and only if for every subset $S \subseteq M$,
\begin{equation}
\cardinal{S\cap f({J_\sfH})} + w_{S,S} \leq (2-\delta)|S|.  \label{equation:S-not-bad-for-f}
\end{equation}
\end{restatable}

\begin{proof}
We define an instance of the single-source single-sink network flow problem as follows.  Construct a directed bipartite graph $H = (A, M, E_H)$, where edges are directed from $A$ to $M$. For every $h, k \in M$ such that $w_{h, k} > 0$, we have a vertex $a_{h,k} \in A$ that is connected to $h$ and $k$. All edges in $E_H$ have infinite capacity.  Now an assignment $f$ is $\delta$-good if and only if we can send flow from $A$ to $M$ in $H$ such that: (1) Each vertex $a_{h,k} \in A$ sends exactly $w_{h,k}$ flow, and (2)
Each machine $i \in M$ receives at most $1_{ i \notin f({J_\sfH})} + 1-\delta$ flow, where $1_{i \notin f({J_\sfH})}$ is 1 if $i \notin f({J_\sfH})$ and $0$ otherwise.
By Hall's theorem, there is a feasible flow if and only if the following holds:
$\sum_{a_{h, k} \in A'}w_{h,k} \leq \sum_{h \in M(A')} \left[1_{ h \notin f({J_\sfH})} + 1-\delta\right], \forall A' \subseteq A$,
where $M(A')$ is the set of vertices in $M$ adjacent to $A'$.  It is easy to see that for every $S \subseteq M$, it suffices to consider the maximal $A'$ with $M(A') = S$. For this $A'$, we have $\sum_{a_{h,k} \in A'}w_{h,k} = w_{S,S}$. Thus, the condition can be rewritten as 
$w_{S,S} \leq (2 - \delta)|S| - \cardinal{S\cap f({J_\sfH})}, \forall S \subseteq M$.
This finishes the proof.
\end{proof}

\subsection{Roadmap of the Proof}

We are now armed to precisely state the ingredients which make up the proof of \Thm{main-2}. 
In \Sec{app-canonical}, we show how to reduce any instance to a canonical instance.  
The precise theorem that we will prove is the following, where $m=|M|$ is the number of machines. 

\begin{restatable}{theorem}{redtocanon}
\label{thm:reducing-to-canonical-instances}\label{thm:1}
Let $\delta > 0, \delta' \in (0, 1)$ and $\rho = 0.6$. Given an instance $\cI$ of the $(1,\epsilon)$-restricted assignment problem with a feasible solution to $\LP(\rho,\delta)$, there is a polynomial time procedure to obtain an $(\frac{m}{\rho \delta},1/\eps,\rho\delta)$-canonical instance $\cI'$ such that any $\delta'$-good assignment for $\cI'$ implies a schedule of makespan at most $(2-\delta'+2\epsilon)$ for $\cI$.
\end{restatable}
\noindent
In \Sec{reducing-p-and-q}, we reduce a canonical instance to one with ``small'' $p$ and $q$. More precisely, we prove the following.

\begin{restatable}{theorem}{redpq}
\label{thm:reducing-p-and-q}\label{thm:2}
For some large enough constant $q_0$ the following is true. Given a $(p,q,\theta)$-canonical instance $\cI$, in polynomial time we can obtain a $(q_0,q_0,\theta + 16\sqrt{\log q_0/q_0})$-canonical instance $\cI'$ such that any $\delta$-good assignment for $\cI'$ is a $(\delta - 16\sqrt{\log q_0/q_0})$-good assignment for $\cI$, for every $\delta  \in (16\sqrt{\log q_0/q_0}, 1)$.
\end{restatable}
\noindent
Finally, in \Sec{fixq}, we show how given a $(q_0,q_0,\theta)$-canonical instance we can `round' it to a $\delta$-good instance where $\delta$ is inverse {\em polylogarithmic} in the $q$ parameter. Observe, from definition of canonical instances, $(1/q-\theta)$-good assignments are trivial.

\begin{restatable}{theorem}{smallpq}
\label{thm:small-q-and-m}
For some large enough constant $C$, every $q_0 \geq 100$ and every $\theta \in [0, \log^{-5}q_0/4C)$, the following is true. Given a $(q_0,q_0,\theta)$-canonical instance $\cI$, there is a polynomial time procedure  to obtain a $(\log^{-5}q_0/C - 4\theta)$-good assignment for $\cI$.
\end{restatable}
\noindent
Assuming the above theorems, the proof of \Thm{main-2} follows easily.
\begin{proofof}{\bf Proof of Theorem \ref{thm:main-2}} Let $C$ be as in Theorem~\ref{thm:small-q-and-m}, and choose $q_0$  such that $400\sqrt{\log q_0/q_0} \leq   \log^{-5}q_0/C$. Let $\delta_0 := \log^{-5}q_0/6C$. Given a feasible solution to $\LP(\rho,\delta_0)$, we convert it to a $(\frac{m}{\rho\delta_0},1/\epsilon, \rho\delta_0)$-canonical instance $\cI$ using \Thm{reducing-to-canonical-instances}. Then using \Thm{reducing-p-and-q}, we obtain a $(q_0,q_0,\rho\delta_0 + 16\sqrt{\log q_0/q_0})$-canonical instance $\cI'$. Given $\cI'$, via \Thm{small-q-and-m}, we obtain a $\delta$-good assignment with $\delta =  6\delta_0  - 4\rho\delta_0 - 64\sqrt{\log q_0/q_0}$. This in turn implies a  $(\delta - 16\sqrt{\log q_0/q_0} = 3.6\delta_0-80\sqrt{\log q_0/q_0})$-good assignment for $\cI$. By choice of parameters, this is a $(2\delta_0 + 2\eps)$-good assignment when $\eps \leq 0.2\delta_0$.
By \Thm{1}, this implies a schedule of makespan $(2-2\delta_0)$ which proves \Thm{main-2}.
\end{proofof}
\noindent
The rest of the paper proves the above theorems in \Sec{app-canonical}, \Sec{reducing-p-and-q}, and \Sec{fixq} respectively which can be read in any order.

\section{Reduction to Canonical Instances}

\label{sec:app-canonical}
\label{app:app-canonical}
This section is devoted to proving \Thm{1}.
\redtocanon*

Let $x$ be any feasible solution for $\LP(\rho = 0.6,\delta)$.  The solution $x$ defines the following weighted bipartite graph $H = (M, J, E_H, x)$: if $x_{i,j} > 0$ for some $i \in M, j \in J$,  there is an edge $(i, j) \in E_H$ of weight $x_{i,j}$. We will create the desired canonical instance $\cI'$ by performing the following sequence of transformation steps. 

\subsection{Processing Heavy Jobs}
Without loss of generality, we can assume $H[M \cup {J_\sfH}]$ is a forest. Indeed, if there is an even cycle in the sub-graph, we can \emph{rotate} the cycle as follows. Color the edges of the cycle alternately as red and black. Uniformly decrease the $x$ values of red edges and increase $x$ values of the black edges. Observe that Constraints~\eqref{CLP:job-covered} and \eqref{CLP:machine-big-job} remain satisfied, and Constraints \eqref{CLP:zi-atmost-1}, \eqref{CLP:machine-small-job}  and \eqref{CLP:compact} are untouched since $z_i$'s and $x_{i,j}$'s for light jobs $j$ did not change.  Apply the operation until the $x$-value of some edge in the cycle becomes $0$.  By applying this operation repeatedly, we can guarantee that the graph $H[M \cup {J_\sfH}]$ is a forest. 
Some heavy jobs $j$ may be completely assigned to a machine $i$; in this case, the edge $(i,j)$ forms a two-node tree, since $z_i \leq 1$. We call such trees trivial.

We now further modify the instance so that each connected component in $H[M \cup {J_\sfH}]$ is a star, with center being a heavy job, and leafs being machines.  Consider any nontrivial tree $\tau$ in the forest $H[M \cup {J_\sfH}]$. We root $\tau$ at an arbitrary heavy job.  If the weight $x_{i,j}$ between any heavy job $j$ in $\tau$ and its parent machine $i$ is at most $1/2$, we remove the edge $(i, j)$ from $\tau$.  After this operation, $\tau$ is possibly broken into many trees.

Now focus on one particular such tree $\tau'$. Note the following facts about $\tau'$: (i) $\tau'$ is rooted at a heavy job $j^*$; (ii) every machine $i$ in $\tau'$ has either 0 or 1 child since $x_{i,j} > 1/2$ for any child $j$ of $i$ in $\tau'$ and \eqref{CLP:zi-atmost-1} holds; (iii) all leaves are machines since a heavy job can only be partially assigned to its parent.  Thus, in $\tau'$, the number of heavy jobs is exactly the number of non-leaf-machines plus 1. 
\begin{lemma}
\label{lem:sum-z-is-large}
Let $L$ be the set of leaf-machines in $\tau'$. Then $\sum_{i \in L}z_i \geq 1/2$.
\end{lemma}
\begin{proof}
Suppose there are $t$ heavy jobs in the tree $\tau'$. Since we may remove an edge of weight at most $1/2$ connecting the root of $\tau'$ to its parent in $\tau$, we have $\sum_{i \in M(\tau')}z_i \geq t-1/2$, where $M(\tau')$ is the set of machines in $\tau'$.  Since $z_i \leq 1$ for each $i \in M(\tau') \setminus L$ and $|M(\tau') \setminus L| = t-1$, we have $\sum_{i \in L}z_i \geq t-1/2 - (t-1) = 1/2$.
\end{proof}

We assign heavy jobs in $\tau'$ to machines in $\tau'$ as follows. Each non-leaf machine of $\tau'$ is guaranteed to be assigned a heavy job.  There is one extra heavy job left, and we assign it to a leaf machine.  The following lemma shows that any leaf-machine can yield to a valid assignment for the heavy jobs.

\begin{lemma}
\label{lemma:choose-any-leaf-machine}
\label{lem:choose-any-leaf-machine}

Let $i$ be any leaf-machine in $\tau'$. There is a valid assignment of heavy jobs in $\tau'$ to machines in $\tau'$ such that
\begin{enumerate}[itemsep=-2pt]
\item Any non-leaf-machine is assigned exactly one heavy job;
\item $i$ is assigned exactly one heavy job;
\item No heavy jobs are assigned to other leaf-machines.
\end{enumerate}
\end{lemma}
\begin{proof}
Focus on the path from the root of $\tau'$ to the leaf-machine $i$.  We assign each heavy job in this path to its child in the path.  For all the other heavy jobs, we assign them to their parents.  It is easy to see this assignment satisfies all the conditions.
\end{proof}
We now create a new set of heavy jobs to \emph{replace} the heavy jobs in $\tau'$.  For each non-leaf machine $i$, we create a new heavy job $j$ with $M_j = \set{i}$. We also create a new heavy job $j$ with $M_j$ being the set of leaf machines.  By \Lem{choose-any-leaf-machine}, a valid assignment for the new machines implies a valid assignment for the original machines.  Notice that new created heavy jobs $j$ have disjoint $M_j$. This is true even if we consider the new jobs for all trees $\tau'$ as the machines in these trees are disjoint. For every new created heavy job $j$, we have $\sum_{i \in M_j} z_i \geq 1/2$: if $M_j = \set{i}$ for a non-leaf machine $i$, then $z_i > 1/2$ as the weight of edge from $i$ to its child has weight at least $1/2$; if $M_j$ is the set of all leaves, then by \Lem{sum-z-is-large}, $\sum_{i \in M_j}z_i \geq 1/2$. 

We have created a new set $J'_\sfH$ of heavy jobs for $\cI$ and the sets $\set{M_j : j \in J'_\sfH}$ are disjoint. An assignment of these big jobs imply an assignment of $J_\sfH$ via \Lem{choose-any-leaf-machine}. From now on we let $J_\sfH = J'_\sfH$ and only consider the set of new heavy jobs. Thus, Property~\ref{property:ci-M} is satisfied.  

Since we haven't modified $x_{i,j}$ and $z_i$ for $i \in M$ and $j\in J_\sfL$, Constraint~\eqref{CLP:zi-atmost-1}, \eqref{CLP:machine-small-job} and \eqref{CLP:compact} are satisfied. Constraint~\eqref{CLP:job-covered} is satisfied for light jobs $j \in J_\sfL$. We did not define $x_{i,j}$'s for the new created heavy jobs $j \in J_\sfH$ and thus Constraint~\eqref{CLP:job-covered} for heavy jobs $j \in J_\sfH$ and Constraint~\eqref{CLP:machine-big-job} are meaningless and henceforth will be ignored. 

We now scale down $z_i$ by a factor of $1 - \rho = 0.4$ for all machines in $i \in M$, then Constraint~\eqref{CLP:compact} is strengthened to
\begin{equation}
z_i + x_{i,j} \leq 1, \qquad \forall j\in J_\sfL, i\in M.
\label{equ:CLP-strengthened-compact}
\end{equation}
For every $j \in J_\sfH$, we have $\sum_{i \in M_j}z_i \geq 0.5(1-\rho) = 0.2$. Every $z_i$ is between $0$ and $0.4$. Moreover, if for some $j \in J_\sfH$ and some $i \in M_j$ we have $z_i = 0$, we remove $i$ from $M_j$. Then, Property~\ref{property:ci-z} holds and Property~\ref{property:canonical-instance-big-job-covered} holds with $\theta = 0$.

\subsection{Processing Light Jobs}
Now let $H$ be the weighted bipartite graph between $M$ and $J_\sfL$. In this step, we make sure that each light job is fractionally assigned to exactly $2$ machines.   To achieve this, perform the rotation operations to cycles in $H$, as we did before for heavy jobs. Note that the rotation preserves the sum $\epsilon \sum_{j \in \cap J_\sfL: i \in M_j} x_{i,j}$. In order to maintain Condition~\eqref{equ:CLP-strengthened-compact}, we may not be able to break all cycles in $H$.  We say an edge $(i, j)$ in $H$ is tight if $x_{i,j} + z_i = 1$.  We can perform rotation operations so that the non-tight edges form a forest.  Also, since $z_i \leq 0.4$ for all machines $i$, $x_{i,j} \geq 1-0.4= 0.6$ for a tight edge $(i,j)$.  Thus, each light job $j$ is incident to at most 1 tight edge.

For each non-singleton tree $\tau$ in the forest formed by the non-tight edges,  we root $\tau$ at an arbitrary light job and {\em permanently} assign each {\em non-leaf} light job in $\tau$ arbitrarily to one of its children. These light jobs are then removed from $J_\sfL$. 
Notice that each machine can get permanently assigned at most 1 light job during this process.  Each unassigned light job in the tree $\tau$ is incident to exactly one non-tight edge (since it was a leaf).

Therefore, each remaining light job $j$ must be one of the following. First, $j$ can be completely assigned to some machine $i$ (thus $x_{i, j} = 1$ and $z_i = 0$), then, we say $j$ is of type-$(i, i)$. Second, $j$ maybe incident to two edges, one tight, the other non-tight. Let $(k, j)$ be a tight edge and $(h, j)$ be the other edge; then we say $j$ is of type-$(h, k)$.   This lets us define the matrix $w$: we let $w_{h,k}$ be the total load of light jobs of type-$(h, k)$, or equivalently, $\eps$ times the number of light jobs of type-$(h, k)$.  For every light job $j$ of type-$(h,k)$, $h\neq k$, we have $x_{h,j} = z_k$ and $x_{k,j} = 1- z_k$.   Thus $w$ satisfies Property~\ref{property:ci-w} and Property~\ref{property:canonical-instance-light-load-large} with $q = 1/\eps$.  Property~\ref{property:canonical-instance-load-small} holds with $\theta = \rho\delta$ as the $z_h + \sum_{k\in M}w_{k, h}(1-z_h) + \sum_{k \in M}w_{h,k}z_k$ is exactly the total fractional load assigned to $h$ which is at most $1+\rho\delta$ by \eqref{CLP:machine-small-job}

Property~\ref{property:canonical-instance-z-i-large} holds for a sufficiently large number $p = \exp(\poly(n))$ as each $z_i$ can be represented using polynomial number of bits.  However, we would like to start with $p = m/(\delta\rho)$, where $m$ is the number of machines. If $0 < z_i < \rho\delta/m$ for some $i \in M_j$, we change $z_i$ to $0$ and remove $i$ from $M_j$. Then, Property~\ref{property:canonical-instance-big-job-covered} still holds for $\theta = \rho\delta/m \times m = \rho\delta$ as there are at most $m$ machines. Thus, our canonical instance is $(m/\rho\delta, 1/\eps, \rho\delta)$-canonical.
This ends the proof of \Thm{1}.


\section{Reducing Parameters \texorpdfstring{$p$}{p} and \texorpdfstring{$q$}{q} in Canonical Instances}\label{sec:reducing-p-and-q}
This section is devoted to proving \Thm{2}. The proof is analogous to a similar theorem proved by Feige~\cite{Fei08} for max-min allocations, and 
therefore we only provide a sketch in the main body. 
All omitted proofs from this section can be found in \App{reducing-p-and-q}.
\redpq*

Using the characterization of $\delta$-good assignment given in \Thm{alpha-good-equivalent-to-no-bad-set}, we define a {\em $\delta$-witness} 
as a pair of sets which rules out any $\delta$-good assignment.

\begin{definition}[$\delta$-witness]
A pair $(S, T)$ of subsets of machines is called a {\bf $\delta$-witness} if $T \subseteq S$ and 
\begin{equation}
\cardinal{T} +  w_{S, S} > (2-\delta)|S|. \label{equation:witness}
\end{equation}
Moreover, we call a $\delta$-witness  $(S, T)$ {\bf connected} if $S$ is (weakly) connected in the light load graph $G_w$.
\end{definition}
\noindent

\begin{restatable}{claim}{connectedwitness}
\label{claim:witness-implies-connected-witness}
\label{clm:witness-implies-connected-witness}
If $(S, T)$ is a $\delta$-witness, then there is a connected $\delta$-witness $(\tilde S, \tilde T)$, with $\tilde S \subseteq S$ and $\tilde T \subseteq T$.
\end{restatable}

\begin{claim}
\label{claim:alpha-good-equivalent-to-no-witness}
\label{clm:alpha-good-equivalent-to-no-witness}
$f$ is a $\delta$-good assignment iff for every connected $\delta$-witness $(S, T)$ of $w$, $T \not \subset f({J_\sfH})$.
\end{claim}

Now, we prove two main lemmas for our algorithm that alternatively reduce $q$ and $p$.   Let $\cI$ be a $(p, q, \theta)$-canonical instance. If $q \geq \max\set{p, q_0}$, \Lem{reduce-q} reduces it to a $(p, q/2, \theta')$-canonical instance; if $p \geq \max\set{q, q_0}$, \Lem{reduce-p} reduces it to a $(p/2, q, \theta')$-canonical instance. 

\begin{restatable}{lemma}{reduceq}
\label{lemma:reduce-q}\label{lem:reduce-q}
 Let $\cI = \left(\set{M_j : j \in {J_\sfH}}, w, z \right)$ be a $(p, q, \theta)$-canonical instance. Assume $q \geq \max\set{p, q_0}$.  Then, we can find in polynomial time a $(p, q', \theta')$-canonical instance $\cI' = \left(\set{M_j:j \in {J_\sfH}}, w', z\right)$, such that any $\delta'$-good assignment $f$ for $\cI'$ is $\delta$-good for $\cI$, where $q' = q/2, \theta' = \theta + 8\sqrt{\log q/q}, \delta' = \delta + 8\sqrt{\log q/q}$.
\end{restatable}

The proof follows from an application of asymmetric LLL. We want that Property~\ref{property:canonical-instance-light-load-large} holds for $q'$.  We apply the following natural procedure. For each $(h, k)$ such that $0 < w_{h,k} < 1/q' = 2/q$, we change $w_{h,k}$ to $1/q'$ with probability $q'w_{h,k}$ and to $0$ with probability $1-q'w_{h,k}$. We need to show that Property~\ref{property:canonical-instance-load-small} holds. Also, we need to show that any $\delta$-witness in the original instance is also a $\delta'$-witness in the new instance. We apply the asymmetric LLL (\Thm{lll}) to show that all these properties can hold.  The idea is that a bad event depending on many other bad events must have a smaller probability.  The detailed proof is in \App{reducing-p-and-q-1}.

\begin{restatable}{lemma}{reducep}
\label{lemma:reduce-p}\label{lem:reduce-p}
 Let $\cI = \left(\set{M_j : j \in {J_\sfH}}, w,  z\right)$ be a $(p, q, \theta)$-canonical instance, where $p \geq \max\set{q, q_0}$.  We can find in polynomial time a $(p', q, \theta')$-canonical instance $\cI' = \left(\set{M_j : j \in {J_\sfH}}, w, z'\right)$ such that any $\delta$-good solution $f$ for $\cI'$ is also $\delta$-good for $\cI$, where $p' =  p/2, \theta' = \theta + 8\sqrt{\log p/p}$.
\end{restatable}

This lemma is by symmetric LLL. To guarantee that each positive $z_i$ has $z_i \geq 1/p' = 2/p$, we apply the following natural process: if $1/p \leq z_i < 1/p'$, we change $z_i$ to $1/p'$ with probability $p'z_i$, and to $0$ with the probability $1-p'z_i$. All bad events in the proof are local; they only depend on a few variables.  Thus, a symmetric LLL suffices to prove the lemma. The detailed proof is in \App{reducing-p-and-q-2}.

To complete the proof of \Thm{2}, we apply \Lem{reduce-q} and \Lem{reduce-p} repeatedly till we obtain a $(q_0, q_0, \theta)$-canonical instance $\cI = (\set{M_j : j \in J_\sfH}, w, z)$, where $\theta \leq \rho\delta_1 + 16\sqrt{\log q_0/q_0}$
with the guarantee that a $\delta$-good solution to $\cI$ implies a $(\delta - 16\sqrt{\log q_0/q_0})$-good solution to the original instance. 

\section{Solving Canonical Instances with Small Values of \texorpdfstring{$p$}{p} and \texorpdfstring{$q$}{q}}
\label{sec:fixq}

This section is devoted to proving \Thm{small-q-and-m}.
\smallpq*

\noindent For convenience, we will make $\theta = 0$ by scaling down the light load matrix $w$ by a factor of $\frac{0.6}{0.6+\theta}$. After this operation, Property~\ref{property:canonical-instance-load-small} will hold with $\theta = 0$ as we have $z_h \leq 0.4$. Let $q = (0.6+\theta)q_0/0.6$; then the new instance is $(q, q, 0)$-canonical except that Property~\ref{property:canonical-instance-big-job-covered} only holds with right side being $0.19$ instead of $0.2$ (as $\theta \leq .01$ for large enough $C$).  At the end, we can scale up light loads by $(0.6+\theta)/0.6$. As each machine is assigned strictly less than  $2$  units of total load, the scaling will increase the light load on each machine by at most $\theta/0.6 \times 2 \leq 4\theta$.  

Thus, we can assume $\theta = 0$  and focus on a $(q, q, 0)$-canonical instance $\cI = (\set{M_j : j \in J_\sfH}, w, z)$ from now on.  With $\theta = 0$, Property~\ref{property:canonical-instance-load-small} implies that $w_{M,h} \leq 1$ for every $h \in M$.

 Given an assignment $f:J_\sfH \to M$, for convenience we use $X = f(J_\sfH)$ to denote the set of machines that are assigned heavy jobs. We define the concept of a `boundary' of a set which will be crucial in what follows.

\def\bnd{{\tt bnd}}
\begin{definition}[Boundary of a set]
Given a subset $S$ of machines, we define its boundary as 
$$\bnd(S) = \sum_{h\in S} \sum_{k\notin S}  \left( w_{k,h}(1-z_h) + w_{h,k}z_{k} \right).$$ 
\end{definition}

\begin{definition}
Let the \emph{deficiency} of a machine $h \in M$ be $\phi_h = 1 \!-\! z_h\! -\! \sum_{k\in M} \big(w_{k,h}(1-z_h) + w_{h,k}z_k\big)$. The deficiency of a subset $S \subseteq M$ is $\phi(S) = \sum_{h \in S}\phi_h$. 
\end{definition}

Thus, $\phi_h \geq 0$ measures how far away Property~\ref{property:canonical-instance-load-small} is from being tight.  With the definition, we can rewrite the condition for $\delta$-good assignments.  From \Thm{alpha-good-equivalent-to-no-bad-set}, $f$ is $\delta$-good iff for every $S\subseteq M$, we have 
\begin{equation}
\label{eq:farkas}
w_{S,S} + |S\cap X| \leq (2-\delta)|S|.
\end{equation}

Adding the definition of $\phi_h$ for every $h \in S$, we get that  $\phi(S) + z(S) + w_{S, S} + \bnd(S)  = |S|$.

\ifdefined\CR
The left hand side of \eqref{eq:farkas} is
\begin{align*}
&\quad |S| + \cardinal{S \cap X} - \phi(S) -z(S)- \bnd(S) \\
&=2|S| - |S \setminus X| - (\phi(S)  + z(S) + \bnd(S)).
\end{align*}
\else
\begin{align*}
\text{LHS of \eqref{eq:farkas}}  &=  |S| + \cardinal{S \cap X} - \phi(S) -z(S)- \bnd(S) \\
&=2|S| - |S \setminus X| - (\phi(S)  + z(S) + \bnd(S)).
\end{align*}
\fi

Thus,  $f$ is $\delta$-good iff for every $S \subseteq M$ we have 
\begin{align}
\label{equ:good-assignments-condtion1}
\cardinal{S \setminus X} + \phi(S) + z(S) + \bnd(S) \geq \delta |S|.
\end{align}

We fix $\delta_0 \in (0, 0.001)$ whose value will be decided later. We say a machine $h \in M$ is {\bf green} if $\phi_h + z_h \geq \delta_0$  and {\bf red} otherwise.  Let $R$ be the set of red machines.  

To check if $f$ is $\delta$-good, we can check the Inequality~\eqref{equ:good-assignments-condtion1} for every $S \subseteq M$.   With some condition on $X$ and some loss on the goodness, we only need to check the above condition for $S \subseteq X \cap R$.  To be more specific,

\begin{lemma}
\label{lem:reduce-to-star}
Let $c_1 \geq 1, \delta\in (0, 1)$.  Suppose $f : J_\sfH \to M$ is a valid assignment with $X=f(J_\sfH)$ satisfying
\begin{properties}
\item for every $h \in M$, we have $\sum_{k\in X\cap R} w_{h,k} \leq c_1 \log q$; \label{property:weight-to-x-small}
\item for every subset $T \subseteq X \cap R$,
$\phi(T) + z(T) + \bnd(T)  \geq \delta|T|$. \label{eq:desire}
\end{properties}
Then $f$ is a $\frac{\delta\delta_0}{2c_1\log q}$-good assignment.
\end{lemma}

\begin{proof}
 Decompose $S = P \cup Q$ where $P = S \cap X \cap R, Q = S \setminus (X \cap R)$. Each machine in $Q$ contributes at least $\delta_0$ to the left-side of Inequality~\eqref{equ:good-assignments-condtion1}. If $|Q| \geq \delta|S|/(2c_1 \log q)$, then the inequality with $\delta$ replaced with $\frac{\delta\delta_0}{2c_1\log q}$ holds trivially.  

  So, assume $|Q| < \delta|S|/(2c_1 \log q)$. 
 For large enough $q$, we get $|P| > 0.99|S|$. 
  By Condition~\ref{eq:desire} for $T$ being $P$,  we have $\phi(P) + z(P) + \bnd(P) \geq \delta|P|$.  Notice that $\bnd(S) \geq \bnd(P) - \sum_{h \in P, k \in Q}\big(w_{k,h}(1-z_h) + w_{h,k}z_k\big)$.  Notice that for any machine $k\in Q$,  $w(k, P) \leq c_1 \log q$ by Condition~\ref{property:weight-to-x-small} and $w(P, k) \leq w(M, k) \leq 1$. We have $\bnd(S) \geq \bnd(P) - |Q|(c_1 \log q + 1) \geq \bnd(P) - 0.51\delta|S|$. Thus, $\phi(S) + z(S) + \bnd(S) \geq \phi(P) + z(P) + \bnd(P) - 0.51\delta|S| \geq \delta|P| - 0.51\delta|S| \geq 0.48\delta|S|$.
\end{proof}

To prove \Thm{small-q-and-m},  it suffices to find an assignment  which satisfies Properties~\ref{property:weight-to-x-small} and \ref{eq:desire} for suitable $\delta$ and $c_1$. In the remainder of this section, we will focus on finding such an assignment.

\paragraph{Sketch of the Proof.}  Suppose for the time being all the positive $w_{h,k}$'s are very close to $1$. $w_{M, h} \leq 1$  implies that in the light load graph $G_w$ any machine $h$ can have in-degree at most $1$. In other words, $G_w$ is a collection of disjoint cycle-rooted trees. Suppose again there are no cycles, and so $G_w$ is a collection of trees. 

Now consider the random process which allocates job $j\in J_\sfH$ to machine $i$ with probability proportional to $z_i$.  We now describe some bad events such that if none of them occur we satisfy   \ref{property:weight-to-x-small} and \ref{eq:desire}. The first bad event is of course the negation of \ref{property:weight-to-x-small} (corresponding to $\cA$ in what follows). The second  bad event (corresponding to $\cB$ in what follows) occurs if the forest induced by $X \cap R$ contains a connected component larger than $\Theta(\log q)$. 

Note that if these bad events don't occur then any subset $S\subseteq X \cap R$ which does not contain roots of trees in the forest satisfy \ref{eq:desire}   (for $\delta = \Theta(\log^{-1} q)$). This is because every connected component contributes $1$ in-degree to $\bnd(S)$. 

Now suppose $S$ contains the root $r$ as well. Since $r$ has no incoming edges and is red, by \ref{property:canonical-instance-load-small}, it has many out-edges. The third bad event (corresponding to $\cC$ in what follows) occurs if lots of out-neighbors of a machine have been picked in $X \cap R$. If $\cC$'s doesn't occur in addition to the $\cB$'s, then the out-edges of $r$ which exit $S$  contributes the boundary term in \ref{eq:desire}. In summary, if no bad event of type $\cA,\cB,\cC$ occur, then \ref{property:weight-to-x-small} and \ref{eq:desire} are satisfied. The formal statement of this is \Lem{no-bad-events-and-we-are-done}.

To show that with positive probability none of the bad events occur, we invoke the asymmetric  Lovasz Local Lemma. To do so, we need to argue about the dependency structure of the various events.  We use a connection  to Galton-Watson branching processes described by Moser and Tardos to prove that $\cB$ has `good' dependency properties (the concrete statement  is \Lem{galton-watson}.). The algorithm follows from the constructive versions of LLL due to~\cite{MT10,HSS11}.

Till now we have been assuming positive $w_{h,k}$'s are close to $1$. In reality, we divide the edges into two classes: {\em dense}, if $w_{h,k} \geq (c_2\log q)^{-1}$ and {\em sparse} otherwise.  Note that dense edges no longer form trees. However, for each machine which has at least one dense edge coming into it (which are called {\em in-dense} later); we arbitrarily pick one of them and color it {\em red} 
and only count the red edges towards the boundary. 

The reason such a `sparsification' helps is that it decreases the dependence among events leading to the application of LLL. To take care of machines with only sparse in-coming edges, we need another type of bad events  (corresponding to $\cD$ in what follows) whose non-existence implies a large contribution to the boundary for such machines.  This ends the sketch of the proof.  We begin the details with some definitions.

\begin{definition}
An edge $(h,k) \in G_w$ is {\bf dense} if ${w_{h,k}}\geq (c_2\log q)^{-1}$, and {\bf sparse} otherwise, where $c_2$ is large enough constant ($c_2 = 300$ suffices).
\end{definition}

\begin{definition}
A machine $h$ is {\bf in-sparse} if {\em all} incoming edges $(k,h)\!\!\in\!\! G_w$ are sparse. Otherwise, $h$ is  {\bf in-dense}.
\end{definition}

\begin{definition}
A machine $h$ is called {\bf out-dense} if $\sum_{k: (h,k) \text{ is dense}} z_k \geq \frac{1}{c_3}$ where $c_3$ is a large enough constant ($c_3 = 200$ suffices). Otherwise, the machine is called {\bf out-sparse}.
\end{definition}

We are now ready to describe our algorithm for assigning heavy jobs to machines satisfying \ref{property:weight-to-x-small} and \ref{eq:desire}.  Our algorithm starts with a pre-processing of the fractional solution, and then recovers a good integral assignment of heavy jobs from the pre-processed solution using randomized rounding. 
note by $X = f(J_\sfH)$ is the set of machines getting heavy jobs.

\subsection{Pre-processing of the Instance}\label{sec:preprocessing}
For every in-dense, red machine $h \in M$, we arbitrarily select an incoming dense edge $(k, h)$ of $h$.  If $k$ is red and out-sparse, we color the edge $(k, h)$ in  red. 
Every machine has at most one incoming red edge. Moreover, the two end points of a red edge are also red. Then each connected component formed by red edges is either a tree, or a tree plus an edge. Recall $X$ is the set of machines which we will assign heavy jobs and $R$ is the set of red machines.

We want to  ensure that $G_w[X \cap R]$ does not contain any red cycles. That is, for any cycle of red edges in $G_w$, we wish to ensure that at least one of these machines is not assigned a heavy job.  For each heavy job $j$, we now identify a subset $M'_j \subseteq M_j$ such that (i) $z(M'_j) \ge 0.49z(M_j)$, and (ii) the subgraph of $G_w$ induced by $\cup_{j \in J_\sfH} M'_j$ does not contain a red cycle. 

We reduce the task of identifying $M'_j$ to an instance of the generalized assignment problem where red cycles correspond to jobs and groups correspond to machines. If a red cycle $C$ contains two machines from a group $M_j$, we can ignore it since one machine in the cycle will not get a heavy job. So assume $C$ contains at most one job from each group $M_j$. The cost of assigning a red cycle $C$ to a group $M_j$ is $z_h$ if some machine $h \in M_j$ participates in the cycle $C$ and $\infty$ otherwise. Since each red cycle $C$ contains at least two machines(if a red cycle contains one machine $h$, then $w_{h,h} > 0$, implying $z_h = 0$ by Property~\ref{property:ci-w} and thus $h$ is not in any $M_j$), a solution that assigns each $C$ uniformly to groups of all machines contained in $C$, is a feasible fractional solution where the load assigned to $M_j$ is at most $z(M_j)/2$. We can now use the Lenstra-Shmyos-Tardos~\cite{LST90} algorithm to recover an integral assignment of red cycles to groups such that maximum load on any machine is at most the fractional load $z(M_j)/2$ plus the largest job size, which is at most $\max_{h\text{ red}}z_h \leq \delta_0$.  Thus for any group $M_j$, the total $z$-value of machines chosen in the red cycle elimination step is at most $z(M_j)/2 + \delta_0 \leq 0.51z(M_j)/$ since $\delta_0 \leq 0.001 \leq 0.01z(M_j)$.

\subsection{Randomized Assignment of Heavy Jobs and Bad Events}
Now we are ready to describe the randomized algorithm to get the heavy job assignment. For every heavy job $j$, assign it to a machine $f(j) := i \in M'_j$ with probability proportional to $z_i$. Note that the probability $p_i$ a machine $i$ gets a heavy job is at most $p_i \leq z_i/(0.49\cdot 0.19) \leq 11z_i$. We describe the bad events.

\begin{itemize}[itemsep=2pt,parsep=2pt]
\item {Bad events $\cA_h$, $h \in M$.} $\cA_h$ occurs if  $\sum_{k\in X\cap R}w_{h,k} > c_1\log q$. 
Setting $c_1 = 12$ is sufficient.
\item {Bad events $\cB_T$, for a set $T$ of $L=10\log q$ machines connected by red edges.} $\cB_T$ occurs if $T \subseteq X \cap R$.
\item {Bad events $\cC_h$, $h \in M$.} $\cC_h$ occurs if  $\cardinal{\set{k \in X \cap R:(h, k)\text{ is dense}}} > 17c_2\log q$.
\item {Bad events $\cD_h$, $h$ is in-sparse.}  $\cD_h$ 
occurs if $\sum_{k\in X \cap R} [w_{k,h}(1-z_h) + w_{h,k}z_k] > 0.1$.
\end{itemize}

\begin{restatable}{lemma}{chernoff}
\label{lem:probabilities}
The bad events described above have probabilities bounded as follows:
\ifdefined\CR
\begin{align*}
\Pr[\cA_h] \leq q^{-c_1}, &\quad \Pr[\cB_T] \leq \prod_{i \in T} p_i, \\
\Pr[\cC_h] \leq q^{-c_2}, &\quad \Pr[\cD_h] \leq q^{-c_2}.
\end{align*}
\else
\begin{equation*}
\Pr[\cA_h] \leq q^{-c_1}, \quad \Pr[\cB_T] \leq \prod_{i \in T} p_i, \quad \Pr[\cC_h] \leq q^{-c_2}, \quad \Pr[\cD_h] \leq q^{-c_2}.
\end{equation*}
\fi
\end{restatable}	

\begin{proof}
The second inequality is trivial. The remaining follow as easy consequences of \Thm{Chernoff}. For any machine $h$, $\Exp[\sum_{k\in X} w_{h,k}] \leq 11\sum_{k\in M} w_{h,k}z_k \leq 11$ since $p_i\leq 11z_i$ and by Property~\ref{property:canonical-instance-load-small} with $\theta = 0$. Since $w_{h,k} \leq 1$, and $c_1,q$ are large enough, \Thm{Chernoff} implies $\Pr[\sum_{k\in X \cap R} w_{h,k} > c_1\log q] \leq \exp(-c_1\log q)$. 

$\Exp[\sum_{k\in X \cap R} w_{h,k}] \leq 11$ implies for any machine $h$, the expected number of out-neighbors $k\in X \cap R$ such that $(h,k)$ is dense is at most $11c_2\log q$.  Therefore, the probability that $\cC_h$ occurs, by the second inequality of \Thm{Chernoff},  is at most $\exp(-(6/11)^2 \cdot 11c_2\log q /3) \leq q^{-c_2}$.

For any machine $h$, the expected value of $\sum_{k\in X \cap R} w_{k,h}(1-z_h) + w_{h,k}z_k$ is at most $11\delta_0$ since red machines are sampled with probability at most $11z_k \leq 11\delta_0$. If $q$ is large enough, $11\delta_0 < 0.005$.  Since $h$ is in-sparse, each $w_{k,h} < (c_2\log q)^{-1}$, and since $k$ is red, $w_{h,k}z_k \leq (C_0 \log q)^{-1}$.  Therefore $w_{k,h}(1-z_h) + w_{h,k}z_k \leq 2/(c_2\log q)$. By the first inequality of  \Thm{Chernoff}, the probability $\cD_h$ occurs is at most $\exp(-20\cdot 0.1 \cdot c_2\log q/2) = q^{-c_2}$.
\end{proof}

\begin{lemma}\label{lem:no-bad-events-and-we-are-done}
If none of the bad events occur,  then $f$ is $\frac{\delta\delta_0}{2c_1\log q} $-good  for $\delta = (340c^2_2c_3\log^3 q)^{-1}$ and  $\delta_0 = (34c_2c_3\log q)^{-1}$.
\end{lemma}

\begin{proof}
We show that if no bad events occur then both conditions of \Lem{reduce-to-star} hold. In fact, since $\cA_h$ doesn't occur for any $h$, we get \ref{property:weight-to-x-small} holds.  Fix a subset $T \subseteq X\cap R$.  We now prove $\phi(T) + z(T) + \bnd(T) \geq \delta|T|$.   This is done by careful accounting.

Focus on an in-sparse machine $h \in T$.  Since $\cD_h$ doesn't occur, we have $\sum_{k \in T}[w_{k,h}(1-z_h)  + w_{h,k}z_k] \leq 0.1$.  By the definition of $\phi_h$ we have $\phi_h + z_h + \sum_{k \notin T}\big(w_{k,h}(1-z_h)  + w_{h,k}z_k\big) \geq 0.9$. Thus, the contribution of $h$ to $\phi(T) + z(T) + \bnd(T)$ is at least $0.9$.


%



The red edges induced by $ST\subseteq X \cap R$ form a forest of rooted trees, by the preprocessing step.   For each machine in the forest, we ask the root of the tree to contribute to $\phi(T) + z(T) + \bnd(T)$. Let $h$ be such a root.

If $h$ is in-sparse, then its contribution is at least $0.9$. Otherwise, $h$ is in-dense and we have selected an dense incoming edge $(k, h)$ in the pre-processing step.  If $k$ is green, then $k \notin T$.  If $k$ is red and out-sparse, then the edge $(k, h)$ is red and thus $k \notin T$.  In either case, the contribution of $h$ is at least $(1-z_h)/(c_2\log q) \geq (1-\delta_0)/(c_2 \log q) \geq 0.99/(c_2 \log q)$. 

It remains to consider the case $k$ is red and out-dense, and $k \in T$. In this case, we ask $k$ to contribute to $\phi(T) + z(T) + \bnd(T)$.   Let $T' = \{k' \in T :(k, k')\text{ is dense}\}$. Since  $T \subseteq X \cap R$ and $\cC_k$ does not happen, we have $|T'| \leq 17c_2 \log q$. $z(T') \leq 17c_2\delta_0\log q$ as all machines in $T'$ are red.
\ifdefined\CR
\begin{align*}
\sum_{k' \notin T}w_{k,k'}z_k' &\geq  \frac{1}{c_2\log q} \sum_{ k' \notin T:(k, k') \text{ dense}}z_\ell\\ 
&\geq \frac{1}{c_2\log q}\left(\frac{1}{c_3}-17c_2\delta_0\log q\right).
\end{align*}
\else
\begin{align*}
\sum_{k' \notin T}w_{k,k'}z_k' &\geq  \frac{1}{c_2\log q} \sum_{ k' \notin T:(k, k') \text{ dense}}z_\ell
\geq \frac{1}{c_2\log q}\left(\frac{1}{c_3}-17c_2\delta_0\log q\right).
\end{align*}
\fi
The quantity is at least $1/(2c_2c_3\log q)$ since  $\delta_0 = (34c_2c_3\log q)^{-1}$.

We count the number of times each machine is asked to contribute.  Since  $\calB$ events do not happen, every root $h$ is asked at most $L$ times.  Since $\cC_k$ does not happen, every $k$ is asked by at most $17c_2\log q$ roots $h$. Overall,  we have proved the lemma for $\delta = \big(2c_2c_3\log q\cdot L \cdot 17c_2\log q\big)^{-1} = \big(340c_2^2c_3\log^3q\big)^{-1}$.

\end{proof}

\subsection{Applying the Asymmetric LLL}
In this section, we show via LLL that no bad event occurs with positive probability. Using the results in~\cite{MT10,HSS11}, we get a polynomial time procedure to obtain an assignment such that no bad events occur. \Lem{no-bad-events-and-we-are-done} proves \Thm{small-q-and-m} for $C = 2^{15}c_1c_2^3c_3^2\log^5 q$.  

We assign each bad event the following $\x$ values: $\x(\cE) = 2\Pr[\cE]$ for any bad event. The key is arguing about the dependence structure of these events which we undertake next by defining the notion of relevant machines for each event.
For any event $\cA_h$,  the relevant subset of machines is the set $\Gamma(\cA_h) = \{k: (h,k)\in G_w\}$ of out-neighbors of $h$.  For any event $\cB_T$,  the relevant subset of machines is the set $\Gamma(\cB_T) = \{k: k\in T \}$. For any event $\cC_h$,  the relevant subset of machines is the set $\Gamma(\cC_h) = \{k: (h,k) \in G_w \text{ and } (h,k)  \text{ is heavy} \}$ of heavy out-neighbors of $h$.  For any event $\cD_h$,  the relevant subset of machines is the set $\Gamma(\cD_h) = \{k: (h,k) \in G_w \text{ or } (k,h) \in G_w \}$ of in-neighbors and out-neighbors of $h$. 
By the facts that $w_{M, h}\leq 1$ , and {\em that all positive $z_k$ and $w_{k,h}$ are at least} $\frac{1}{q}$, we get that
$\max_h \{ |\Gamma(\cA_h)|, |\Gamma(\cC_h)|, |\Gamma(\cD_h)| \} \leq 2q^2$. 
\def\group{{\tt group}}
For a machine $h\in M$, we let $\group(h)$ denote the set $M_j$ where $h\in M_j$ if it exists; otherwise $\group(h) = \set{h}$. 
\begin{definition}
Two sets $S$, $T$ of machines are {\bf group-disjoint} if no machine in $S$ is in the same group with a machine in $T$. That is, for any $u\in S, v \in T$, $\group(u) \neq \group(v)$.
\end{definition}

\begin{claim}
An event $\cA_h/\cB_T/\cC_h/\cD_h$ is independent of $\cA_k/\cB_{T}/\cC_k/\cD_k$ if the relevant subset of machines for these events are  group-disjoint.
\end{claim}

The main non-trivial lemma is the following.
\begin{restatable}{lemma}{galtonwatson}
\label{lem:galton-watson}
For every red machine $h \in R$, we have $\prod_{S: h\in S} \left(1-\x(\cB_S)\right) \geq \exp(-4/3^L)$ where $S$ is over all sets of size $L$ containing $h$ and connected by red edges.
\end{restatable}
\begin{proof} 
To argue about the probability of a connected set $S$ being chosen so that $h\in S$,  focus on the red machines $\bigcup_{j \in J_\sfH}M'_j$ and the red edges induced on these machines.  The graph is a directed forest. We remove the directions, and focus on the connected component containing $h$. Root this tree at $h$ and let $\Lambda(v)$ be the set of children of $v$ for any $v$ in this rooted tree. Consider a branching process which selects $h$ with probability $4p_h$ (recall $p_h$ is the probability a machine $h$ gets a heavy job), and if $h$ is chosen, pick each child $k \in \Lambda(v)$ with probability $4p_k$, and continue thus.
When this process terminates, we end up with some connected set $S$. The probability that we get a specific set $S$  is precisely 
\ifdefined\CR
\begin{align*}
&\quad \pi_S := \Pr_{\text{branching process}} [S] \\
&= (1-4p_h) \prod_{k \in S}\left(\frac{4p_k}{1-4p_k}\prod_{u \in \Lambda(k)} \left(1-4p_u\right) \right).
\end{align*}
\else
$$\textstyle \pi_S := \Pr_{\text{branching}} [S] = (1-4p_h) \prod_{v\in S}\left(\frac{4p_k}{1-4p_k}\prod_{u \in \Lambda(k)} \left(1-4p_u\right) \right).$$
\fi
Since $p_u \leq 11\delta_0$ as $u$ is red, we have $1-4p_u \geq \exp(-5p_u)$. Therefore, $\prod_{u \in \Lambda(k)} \left(1-4p_u\right)\geq \exp\left(-5\sum_{u\in \Lambda(k)} p_u\right)$. If $k$ is out-dense, then $k$ has degree 1 thus has no children; the quantity is 1.
If $k$ is out-sparse, then $\sum_{u\in \Lambda(k)} p_u \leq 11\sum_{u\in \Lambda(k)} z_u \leq 11/c_3$. As $c_3 \geq 200$, we get that  $\exp\left(-5\sum_{u\in \Lambda(k)} p_u\right) \geq 3/4$. This implies $\pi_S \geq \prod_{k\in S}(3p_k) \geq 3^L\Pr[B_S]$.
Since $\sum_{S: |S| = L} \pi_S \leq 1$ (the branching process leads to one set $S$), we get $\sum_{S \ni h} \Pr[B_S] \leq \frac{1}{3^L}$.
Now, 
$\prod_{S: h\in S} \left(1-\x(B_S)\right) \geq \exp(-2\sum_{S:h\in S} \x(B_S)) = \exp(-4\sum_{S:h\in S} \Pr[B_S]) \geq \exp(-4/3^L)$.
\end{proof}
\noindent
To check the conditions of the asymmetric LLL is now just book keeping. 
For any bad event $\cB_S$ we have 
\ifdefined\CR
\begin{align*}
&\prod_{T:\cB_T\sim \cB_S}\left(1-\x(\cB_T) \right) \geq  \prod_{h \in S} \prod_{k \in \group(h)} \prod_{T \ni k} (1-\x(\cB_T))\\
&\qquad \geq \exp(-\frac{4}{3^L}|S||M_j|) \geq \exp(-\frac{4qL}{3^L}),
\end{align*}
\else
$$\prod_{T:\cB_T\sim \cB_S}\left(1-\x(\cB_T) \right) \geq  \prod_{h \in S} \prod_{k \in \group(h)} \prod_{T \ni k} (1-\x(\cB_T)) \geq \exp(-\frac{4}{3^L}|S||M_j|) \geq \exp(-\frac{4qL}{3^L}),$$
\fi
where $j$ is the heavy job which can go to $h$, and since $z_h \geq 1/q$, we have $|M_j| \leq q$. 
Since $L = 10\log q$, the RHS is at least $0.9$ whenever $q > 10$.

Fix an event $\cE_h$ where $\cE \in \{\cA,\cC,\cD\}$. 
Let us calculate upperbound the product
\ifdefined\CR
\begin{align*}
&\quad \prod_{S:\cB_S\sim \cE_h}\left(1-\x(\cB_S)\right)\\
&\geq \prod_{k \in \Gamma(\cE_h)} \prod_{k' \in \group(k)} \prod_{S: k'\in S} (1-\x(\cB_S))\\
&\geq \exp(-\frac{4}{3^L}|\Gamma(\cE_h)||\group(k)|) \geq \exp(-\frac{4q^3}{3^L}),
\end{align*}
\else
$$\prod_{S:\cB_S\sim \cE_h}\left(1-\x(\cB_S)\right) \geq \prod_{k \in \Gamma(\cE_h)} \prod_{k' \in \group(k)} \prod_{S: k'\in S} (1-\x(\cB_S)) \geq \exp(-\frac{4}{3^L}|\Gamma(\cE_h)||\group(k)|) \geq \exp(-\frac{4q^3}{3^L}),$$
\fi
where $j$ is the heavy job which contains $k'$ in its $M_j$. Since $L > 10\ln q$, the RHS is at least $0.9$ since $q > 10$.
Similarly, 
\ifdefined\CR
\begin{align*}
&\quad\prod_{h: \cE_h \sim \cB_S} \left(1 - \x(\cE_h)\right)\\
&\geq \prod_{k\in S} \prod_{k'\in \group(k)} \prod_{h: k'\in \Gamma(\cE_h)} \left(1 - \x(\cE_h)\right)\\ 
&\geq \exp\left(-2\sum_{k\in S}\sum_{k'\in \group(k)}\sum_{h: k'\in \Gamma(\cE_h)} \x(\cE_h)\right).
\end{align*}
\else
$$\prod_{h: \cE_h \sim \cB_S} \left(1 - \x(\cE_h)\right) \geq \prod_{k\in S} \prod_{k'\in \group(k)} \prod_{h: k'\in \Gamma(\cE_h)} \left(1 - \x(\cE_h)\right) \geq \exp\left(-2\sum_{k\in S}\sum_{k'\in \group(k)}\sum_{h: k'\in \Gamma(\cE_h)} \x(\cE_h)\right).$$
\fi
The number of terms in the RHS is at most $|S|q^3_0 \leq q^4_0$. Since $\x(\cE_h) = 2\Pr[\cE_h] \leq q^{-6}$ (if $c_2 \geq  25$ and $c_1 \geq 6 $), the RHS is atleast $0.9$ for large enough $q$.

Finally, we have
\ifdefined\CR
\begin{align*}
&\prod_{k:\cE_k\sim \cE_h}\left(1-\x(\cE_k)\right) \geq \prod_{h' \in S_h} \prod_{h'' \in \group(h')} \prod_{k: h''\in S_k} (1-\x(\cC_k))\\
 &\geq \exp\left(-2\sum_{h'\in S_h}\sum_{h''\in \group(h')}\sum_{k: h''\in S_k} \x(\cC_k)\right).
\end{align*}
\else
$$\prod_{k:\cE_k\sim \cE_h}\left(1-\x(\cE_k)\right) \geq \prod_{h' \in S_h} \prod_{h'' \in \group(h')} \prod_{k: h''\in S_k} (1-\x(\cC_k)) \geq \exp\left(-2\sum_{h'\in S_h}\sum_{h''\in \group(h')}\sum_{k: h''\in S_k} \x(\cC_k)\right).$$
\fi
For any machine $h''$, the set $\{k: v\in S_k\}$ is precisely the neighbors of $h''$ (the machines of which $h''$ is a neighbor of are precisely the neighbors of $h''$). Therefore, the number of terms in the summation is at most 
$q^2\cdot q\cdot q^2 \leq q^5$. Since $\x(\cC_k) \leq q^{-6}$, we get that the RHS is at least $0.9$ for large enough $q$. In sum, we get that for any event $\cE \in \{\cA_h, \cB_S,\cC_h,\cD_h\}$, we have
\ifdefined\CR
\begin{align*}
&\x(\cE)\prod_{T:\cB_T\sim \cE} \left(1-\x(\cB_T)\right)\prod_{k:\cC_k\sim \cE}(1 - \x(\cC_k))\\
&\quad \times \prod_{k:\cD_k\sim \cE} (1- \x(\cD_k)) \quad >\quad 0.5\x(\cE) \quad = \quad \Pr[\cE],
\end{align*}
\else
$$\x(\cE)\prod_{T:\cB_T\sim \cE}\left(1-\x(\cB_T)\right) \prod_{k:\cC_k\sim \cE}(1 - \x(\cC_k)) \prod_{k:\cD_k\sim \cE} (1- \x(\cD_k)) > 0.5\x(\cE) = \Pr[\cE],$$
\fi
implying that $x$ satisfies the LLL condition. 

Finally note that the number of events of type $\cA,\cC,\cD$ are polynomially many, and given an assignment of heavy jobs, one can easily check if one of the $\cB_T$ occurs or not. Therefore, \Thm{hss-1} applies  
and this completes the proof of \Thm{small-q-and-m}.

\bibliographystyle{plain}
\bibliography{reflist}

\appendix
\section{A \texorpdfstring{$(2-\eps)$}{2-epsilon} algorithm for the \texorpdfstring{$(1,\eps)$}{(1,epsilon)}-restricted assignment problem.}
\label{app:large-eps}

\begin{theorem}
There exists a polynomial time algorithm which returns a $(2-\epsilon)$-approximation to the makespan minimization problem in $(1,\epsilon)$-restricted assignment instances. There exists a polynomial time $11/6 \approx 1.833$-factor algorithm to estimate the optimal makespan in $(1,\epsilon)$-restricted assignment instances.
\end{theorem}
\begin{proof}
We construct a bipartite matching problem. Vertices on the right side correspond to jobs. Suppose $\OPT$ is the optimum makespan.  For each machine $i$, we create $\floor{\OPT}$ heavy slots and  $\floor{\OPT/\eps}$ - $\floor{\OPT}$ light slots. There is an edge between a heavy slot and all jobs that can be assigned to the machine; there is an edge between a light slot and all light jobs that can be assigned to the machine. It is easy that there is a matching that covers all jobs. Each machine gets a total load at most $\floor{\OPT} + \epsilon(\floor{\OPT/\eps} - \floor{\OPT}) = (1-\eps)\floor{\OPT} + \epsilon\floor{\OPT/\eps} \leq (2-\eps)\OPT$.

This gives a $2-\eps$ approximation for the problem. By combining this with the $(5/3+\eps)$-estimation algorithm of Svensson~\cite{Sve11}, we obtain an algorithm that estimates the make span  up to a factor of $\min \set{2-\eps, 5/3+\eps} \leq 11/6$.

\end{proof}

\section{Hardness of \texorpdfstring{$(1,\epsilon)$}{(1,epsilon)}-restricted assignment problem}
\label{app:hardness}
We complement our algorithmic result with the following hardness of approximation result.

\begin{theorem}\label{thm:3}
For any $\epsilon> 0$, 
it is NP-hard to approximate the makespan of the $(1,\epsilon)$-restricted assignment problem to a factor better than $7/6$.
\end{theorem}
\begin{proof}
We reduce from the problem of finding a vertex cover in cubic graphs: there exists parameter $K(n)$ so that it is NP hard to decide whether an $n$-vertex cubic graph has a vertex cover of size $\le K(n)$ or not.
Given a vertex cover instance $G=(V,E)$ on $n$ vertices, we construct an instance of $(P|\gamma|C_{max})$ as follows: we have a machine for every vertex $v\in V(G)$, a set of $n - K(n)$ heavy jobs that can be assigned to any machine,
a set of $\frac{1}{3\eps}$ light jobs $S_e$ for every edge $e=(u,v)\in E(G)$ with job $j\in S_e$ having
$p_j = \eps$ and can be scheduled on machine $u$ or machine $v$.

If $G$ has a vertex cover of size $\le K(n)$, then we can find a schedule of makespan $1$. 
Let $U\subseteq V$ be the vertex cover; allocate all heavy jobs to machines corresponding to $V\setminus U$. For every edge $e = (u,v)$, we are guaranteed one of the end points lies in $U$ and thus doesn't have a heavy job. Allocate all jobs of $S_e$ to that machine. Any machine gets a total small load of at most $1$, and any machine getting a heavy job doesn't get a light job. 

If $G$ doesn't have a vertex cover of size $\le K(n)$, then no matter how the heavy jobs are allocated, there must be an edge $e = (u,v)$ such that both $u$ and $v$ are allocated heavy jobs. The total load on one of these two machines is at least $1 + 1/6 = 7/6$.
\end{proof}

\section{Some Useful Tools}\label{app:useful-tools}
We state below two results that we will frequently utilize in our analysis. 

\begin{theorem}
\label{thm:Chernoff}
Let $Z$ be the sum of independent scalar random variables each individually in range $[0,K]$ and $\mu = \Exp[Z]$. Then for any $\lambda \geq 7$, we have
\[ \Pr[Z \geq \lambda \mu] \leq  e^{-\lambda\mu/K}. \]
For any $\lambda  \in (0, 1)$, we have
\ifdefined\CR
\begin{align*}
\Pr[Z \geq (1+\lambda) \mu] &\leq  e^{-\lambda^2\mu/3K}, \\
\Pr[Z \leq (1-\lambda) \mu] &\leq  e^{-\lambda^2\mu/2K}.
\end{align*}
\else
\[\Pr[Z \geq (1+\lambda) \mu] \leq  e^{-\lambda^2\mu/3K}, \text{ and } \Pr[Z \leq (1-\lambda) \mu] \leq  e^{-\lambda^2\mu/2K}.\]
\fi
\end{theorem}
\begin{proof}
All those bounds are simple application of standard Chernoff bounds.  Let $X$ be the sum of $n$ independent random variables, each take value in $[0, 1]$.  Let $\mu = \E[X]$. Then for every $\delta > 0$, we have 
\begin{equation*}
\Pr[X \geq (1 + \delta)\mu] \leq \left(\frac{e^{\delta}}{(1+\delta)^{1+\delta}}\right)^\mu.
\end{equation*}
For every $\delta \in (0, 1)$, we have 
\begin{equation*}
\Pr[X \leq (1-\delta)\mu] \leq \left(\frac{e^{-\delta}}{(1-\delta)^{1-\delta}}\right)^\mu.
\end{equation*}

To prove the theorem, we can scale the random variables by a factor of $1/K$ and then mean of $Z$ is changed to $\mu/K$. Thus, we can assume $K=1$.  

The first inequality is obtained by setting $\delta = \lambda - 1$ and observing that $e^{\lambda - 1}/\lambda^\lambda \leq e^{-\lambda}$ for $\lambda \geq 7$.  For the second inequality and third inequality, let $\delta = \lambda$. The second inequality holds since $e^{\delta}/(1+\delta)^{1+\delta} = \exp\big(\delta - (1+\delta)\ln(1+\delta)\big) \leq \exp\big(\delta - (1+\delta)\frac{2\delta}{2+\delta}\big) = \exp\left(\frac{-\delta^2}{2+\delta}\right) \leq \exp(-\delta^2/3)$. The third inequality holds since $e^{-\delta}/(1-\delta)^{1-\delta} = \exp\big(-\delta - (1-\delta)\ln(1-\delta)\big) \leq \exp\big(-\delta - (1-\delta)(-\frac{2\delta}{2-\delta })\big) = \exp\left(-\frac{\delta^2}{2-\delta}\right) \leq \exp(-\delta^2/2)$.

\end{proof}
\begin{theorem}[Asymmetric LLL]\label{thm:lll}
Let ${\cal E} = \{ E_1,\ldots, E_m \}$ be a finite collection of (bad) events in a probability space.
For each $E_i$, let $\Gamma(E_i)$ denote a subset of events such that $E_i$ is independent of each event in ${\cal E}\setminus (E_i \cup \Gamma(E_i))$. Then if there exists an assignment 
$\x: {\cal E} \rightarrow (0,1)$ satisfying the property $\Pr[E_i]\leq \x(E_i)\cdot \prod_{E_j \in \Gamma(E_i)} \left(1-\x(E_j)\right)$,
the probability that none of the events in ${\cal E}$ occurs is at least $\prod_i \left(1 - \x(E_i)\right)$.
\end{theorem}

While the version stated above is only an existence statement, the recent work of Moser and Tardos~\cite{MT10}, and Haeupler et al.~\cite{HSS11} has given polynomial-time algorithms for finding a solution that avoids all bad events. We  use the notation $E_j \sim E_i$ to indicate that $E_j \in \Gamma(E_i)$.
Let $V = \set{v_1, v_2, \cdots, v_n}$ be $n$ independent random variables. Let ${\cal E} = \set{E_1, E_2, \cdots, E_m}$ be a finite collection of (bad) events where each $E_i$ only depends on a subset $V_i \subseteq V$ of variables.   Let $\Gamma(E_i) = \set{E_{i'}: i' \neq i, V_i \cap V_{i'} \neq \emptyset}$.  

The Moser-Tardos (MT, henceforth) algorithm does the following: a) Initially sample $v_i$'s independently, and b) until all $E_i$'s are dissatisfied, pick an arbitrary satisfied $E_i$ and resample the $v_j$'s present in $V_i$.
Moser and Tardos~\cite{MT10} showed that if the LLL condition held, the above algorithm terminated in $O\left(\sum_{i=1}^m \x(E_i)(1-\x(E_i))^{-1}\right)$ iterations. 
This suffices for many applications; however there are two issues -- a) $m$ could be superpolynomial in $n$, and b) given a setting of $v_i$'s there may not be an efficient method to detect if a satisfied $E_i$ exists or not.
Haeupler et al.~\cite{HSS11} addressed these issues in the following ways. 

\begin{theorem}(Paraphrasing of Theorem 3.1 in~\cite{HSS11})\label{thm:hss-1}
Suppose the LLL condition holds, and let $\delta := \min_{j} \x(E_i)\prod_{j\sim i} (1-\x(E_j))$. Then the expected number of resamplings of the MT algorithm is at most 
$n\log(1/\delta)\max_i (1-\x(E_i))^{-1}$.
\end{theorem}

The above theorem takes care of situations where the number of events may be superpolynomial; however, `efficient verifiability' occurs, that is, given a setting of $v_i$'s one can detect a satisfied  $E_j$ or assert none hold.
To take care of issue (b) above, Haeupler et al.~\cite{HSS11} modified the MT algorithm as follows. It parametrizes the events 
with a set ${\cal E}' \subseteq {\cal E}$ of {\em core events}. Randomly and independently assign a value to each random variable in $V$.  In each iteration, we check if any bad event $E_i \in {\cal E}'$ happens. If there is such a bad event $E_i \in {\cal E}'$, we resample all variables in $V_i$ and start a new iteration. Otherwise we terminate the algorithm return the current assignment.  

\begin{theorem}(Paraphrasing of Theorem 3.4 in~\cite{HSS11})\label{thm:hss-2}
Suppose there exists an $\eps \in (0,1)$ and assignment $\x:{\cal E} \mapsto (0,1-\epsilon)$ such that a slightly stronger-than-LLL condition holds:
\begin{equation}
\label{eq:HSS}
{\mathrm{Pr}}^{1-\eps}[E_i]\leq \x(E_i)\cdot \prod_{E_j \in \Gamma(E_i)} \left(1-\x(E_j)\right)
\end{equation}
Suppose further, that $\log(1/\delta) \leq \poly(n)$. Then 
\begin{enumerate}
\item For any $p \geq 1/\poly(n)$, the set ${\cal E}' := \{E_i: \Pr[E_i] \geq p\}$ is of size at most $\poly(n)$.
\item With probability $(1-n^{-c})$, the HSS algorithm with core events ${\cal E}'$ terminates after $O(n\log n)$ resamplings and returns an assignment such that no event in $\cal E$ occurs.
\end{enumerate} 
\end{theorem}


\section{Omitted details from \Sec{reducing-p-and-q}}
\label{app:reducing-p-and-q}
\connectedwitness*
\begin{proof}
Consider all the (weakly) connected components of $G_w[S]$(the sub-graph of $G_w$ induced by $S$). There must be some connected component induced by $\tilde S \subseteq S$ such that $\cardinal{T \cap \tilde S} + w_{\tilde S, \tilde S} > (2-\delta)\cardinal{\tilde S}$, since summing up the left side over all connected components $\tilde S$ gives $\cardinal{T} + w_{S,S}$ and summing up the right side gives $(2-\delta)\cardinal{S}$. Thus, $\left(\tilde S, \tilde T =T \cap \tilde S\right)$ is a connected $\delta$-witness.
\end{proof}

\Clm{alpha-good-equivalent-to-no-witness} follows
immediately from \Thm{alpha-good-equivalent-to-no-bad-set} and \Clm{witness-implies-connected-witness}.

\subsection{Proof of \Lem{reduce-q}}
\label{app:reducing-p-and-q-1}

Before the proof of the Lemma, we need one simple claim. 
\begin{claim}
\label{clm:degree-is-small}
For any $(p, q, \theta)$-canonical instance, 
we have ${w_{M,h}} \leq 1.1, {w_{h,M}} \leq 1.1p$ for every $h \in M$. 
\end{claim}

\begin{proof}
Since $z_h + {w_{M,h}}(1-z_h) \leq 1 + \theta \leq 1.05$ and $z_h \leq 1/2$, we have $w_{M,h} \leq \frac{1.05-z_h}{1-z_h} \leq 1.1$.  

Consider any machine $h \in M$. Notice that $\sum_{k\in M} w_{h,k} z_{k} \leq  1 + \theta \leq 1.1$ and $z_{k} \geq 1/p$ if $w_{h,k} > 0$. We have $w_{h,M} \leq \frac{1.1}{(1/p)} = 1.1p$.
\end{proof}

\reduceq*
\begin{proof}

For each pair $(h, k)$ with $0 < w_{h, k} < 1/q' = 2/q$, we let $w'_{h,k} = 1/q'$ with probability $q'w_{h,k}$ and let $w'_{h,k} = 0$ with probability $1-q'w_{h,k}$.  For all other pairs $(h,k)$, we let $w'_{h,k} = w_{h,k}$. Then $\cI' = (\set{M_j:j\in J_\sfH}, w', z)$ is the new canonical instance. 
\begin{enumerate}
\item ${\calA}_h$, for every machine $h \in M$:  ${\calA}_h$ occurs if $z_h+w'_{M,h}(1-z_h) + \sum_{k \in M}w'_{h,k}z_{k} > 1 + \theta'$;
\item ${\calB}_{S, T}$, for every connected $\delta$-witness $(S, T)$ of $\cI$: ${\calB}_{S, T}$ occurs if $(S, T)$ is not a $\delta'$-witness of $\cI'$.
\end{enumerate}

If none of the bad events occur, then $\cI'$ is a $(p,q/2,\theta')$-canonical instance; furthermore, any $\delta'$-good assignment for $\cI'$ must be a $\delta$-good assignment for $\cI$ since
otherwise $\cB_{S,T}$ would occur for some connected $\delta$-witness.
In the rest of the proof, we use LLL to show that none of the bad events occur with positive probability. Using the techniques of ~\cite{MT10,HSS11}, there is a polynomial time procedure which obtains $\cI'$ with the desired property. 

Focus on the quantity $W$ on the left side of the inequality defining $\calA_h$.  All random variables (the $w'_{h,k}$s) in $W$ take value in $\set{0, 1/q'}$ and the coefficient before each random variable in $Z$ is at most 1; moreover, $\E(w'_{h,k}) = w_{h,k}$.
By Property~\ref{property:canonical-instance-load-small}, we have $z_h + {w_{M,h}}(1-z_h) + \sum_{k \in M}w_{h,k}z_{k} \leq 1+\theta < 1.1$.    The Chernoff bound in \Thm{Chernoff} gives that the probability that ${\calA}_h$ occurs is at most $\exp\left(-(\theta' - \theta)^2q'/3.3\right)\leq \exp(-8\log q) = q^{-8}$.  

Now consider the bad event ${\calB}_{S, T}$. Since $(S, T)$ is a $\delta$-witness of $\cI$, we have $\cardinal{T} + {w_{S,S}} > (2-\delta)\cardinal{S}$.   ${\calB}_{S, T}$ occurs if $\cardinal{T} + {w'_{S,S}} \leq (2-\delta')\cardinal{S}$. Again, by Chernoff bound, the probability that ${\calB}_{S, T}$ occurs is at most $\exp\left(-(\delta' - \delta)^2q'\cardinal{S} /4\right) \leq e^{-8(\log q)\cardinal{S}} = q^{-8\cardinal{S}}$.

Now we apply the (asymmetric) LLL. In order to apply LLL, we need to define the $\x$ values for the bad events. 
Define $\x({\calA}_h) = q^{-7}$ and $\x({\calB}_{S, T}) = q^{-7\cardinal{S}}$. 

Focus on some bad event ${\calA}_h$. If ${\calA}_k$ is dependent of ${\calA}_h$, then either $w_{k, h} > 0$, or $w_{h, k} > 0$ and $z_k > 0$. By \Clm{degree-is-small}, the number of events ${\calA}_{k}$ dependent of ${\calA}_h$ is at most $(1.1 + 1.1p)/(1/q) < q^3/4$ since each positive $w_{h,k}$ has $w_{h,k} \geq 1/q$ and $p\leq q$.  We count the number of events ${\calB}_{S, T}$ dependent on ${\calA}_h$ satisfying $\cardinal{S} = t$.  ${\calB}_{S,T}$ is dependent on ${\calA}_h$ only if $h \in S$. Since the degree of vertices in $G_w$ is at most $q^3/4$, and $G_w[S]$ is connected, the number of sets $S$ is at most $(q^3)^t$. \footnote{We can use the same argument as in \cite{Fei08}. Given a graph $G$ of degree $d$ and a vertex $v$, we want to bound the number of induced connected sub-graphs of $s$ containing $v$. Fix an arbitrary spanning tree for the sub-graph and root it at $v$. There are at most $2^{2s} = 4^s$ tree structures: visiting the tree in the DFS order and we only need to specify which $s-1$ edges are forward moves. Given a tree structure, there at at most $d^s$ choices for the tree. Thus the number is bounded by $(4d)^s$.} 
For a fixed $S$, there are at most $2^t$ different sets $T$.  Thus, the number of such dependent events is at most $(q^3)^t \times 2^t \leq q^{4t}$.  This gives, 
\ifdefined\CR
\begin{align*}
&\quad \x\left({\calA}_h\right)\prod_{{\cE}\sim {\calA}_h}\left(1 - \x({\cE})\right)\\
&\geq q^{-7} \left(1 - q^{-7}\right)^{q^3}\prod_{t \geq 1}\left(1 - q^{-7t}\right)^{q^{4t}} \geq q^{-8} \geq \Pr({\calA}_h),
\end{align*}
\else
\[
\x\left({\calA}_h\right)\prod_{{\cE}\sim {\calA}_h}\left(1 - \x({\cE})\right) \geq q^{-7} \left(1 - q^{-7}\right)^{q^3}\prod_{t \geq 1}\left(1 - q^{-7t}\right)^{q^{4t}} \geq q^{-8} \geq \Pr({\calA}_h),
\]
\fi
where the product in the LHS is over all events $\cE$ dependent on $\cA_h$.  

Now consider some bad event ${\calB}_{S, T}$ with $\cardinal{S} = s$. Using a similar counting argument, the number of events ${\calA}_h$ that are dependent on ${\calB}_{S, T}$ is at most $sq^3$ and the number of events ${\calB}_{\tilde S, \tilde T}$ dependent on ${\calB}_{S, T}$ satisfying $\tilde S = t$ is at most $sq^{4t}$.  Thus, 
\ifdefined\CR
\begin{align*}
&\quad \x\left({\calB}_{S, T}\right) \prod_{{\cE}\sim {\calB}_{S, T}}(1 - \x({\cE}))\\
&\geq q^{-7s} \left(1 - q^{-7}\right)^{sq^3}\prod_{t \geq 1}\left(1 - q^{-7t}\right)^{sq^{4t}}\\
& \geq q^{-8s} \geq \Pr({\calB}_{S, T}).
\end{align*}
\else
\[
\x\left({\calB}_{S, T}\right) \prod_{{\cE}\sim {\calB}_{S, T}}(1 - \x({\cE})) \geq q^{-7s} \left(1 - q^{-7}\right)^{sq^3}\prod_{t \geq 1}\left(1 - q^{-7t}\right)^{sq^{4t}} \geq q^{-8s} \geq \Pr({\calB}_{S, T}).
\]
\fi
We have verified that the conditions for the asymmetric LLL, and this completes the proof of the lemma.

To see how the theorems of~\cite{MT10,HSS11} can be applied, note that that there are at most $m$ events of the type $\cA_h$. The events $\cB_{S,T}$ are exponentially many and do not seem to be efficiently verifiable. 
This is where one uses \Thm{hss-2} of~\cite{HSS11}. Note in the above analysis, \eqref{eq:HSS} holds with $\eps = 1/7$. The theorem implies the `core bad events' which have $\Pr[\cB_{S,T}] \geq 1/\poly(m)$, that is, those with  $|S| = O(\frac{\log m}{\log q})$, are at most $m^{O(1)}$. Since these can be enumerated over using a BFS tree, 
we can find a `good' assignment in polynomial time.
\end{proof}

\subsection{Proof of \Lem{reduce-p}}
\label{app:reducing-p-and-q-2}
\reducep*
\begin{proof}
For every $h \in M$ with $ 0 < z_h < 1/p' = 2/p$, we let $z'_h = 1/p'$ with probability $p'z_h$ and let $z'_h = 0$ with probability $1-p'z_h$. For all other machines $h$, we let $z'_h = z_h$.  Note that $\Exp[z'_h] = z_h$. 

To make $(\set{M_j:J \in J_\sfH}, w, z')$ a canonical instance, we need to apply more operations. If some $h \in M_j$ has $z'_h = 0$, we need to remove $h$ from $M_j$. If $z'_k = 0$, for every $h \neq k$, we need to change the $w_{h,k}$ light load of type-$(h, k)$ to load of type-$(k, k)$.  However, these operations do not affect our proof. Thus, we can pretend our new instance is $\cI' = (\set{M_j:J \in J_\sfH}, w, z')$.

Since the definition of a $\delta$-good assignment is independent of $z$,  a $\delta$-good assignment for $\cI'$ is a $\delta$-good assignment for $\cI$. The non-trivial part is to show that $\cI'$ is $(p',q,\theta')$-canonical. The non-trivial properties are  \ref{property:canonical-instance-big-job-covered} and \ref{property:canonical-instance-load-small}. Note that \ref{property:canonical-instance-z-i-large} is satisfied by the construction above and \ref{property:canonical-instance-light-load-large} is untouched.

To this end, consider the following two types of bad events. If none of the bad events occur we are done. Once again, we use LLL to show that none of the bad events occur with positive probability, and the lemma is proven by the theorems of~\cite{MT10,HSS11}.
\begin{enumerate}
 \item ${\calA}_h, h \in M$: ${\calA}_h$ occurs if $z'_h+w_{M,h}(1-z'_h) + \sum_{k \in M}w_{h,k}z'_{k} > 1+\theta'$.
 
\item ${\calB}_j, j \in {J_\sfH}$. ${\calB}_j$ occurs if $\sum_{h\in M_j}z'_h < 0.2 - \theta'$.
\end{enumerate}

Consider the quantity $Z$ on the left side of the inequality defining ${\calA}_h$. Notice that we have $z_h+{w_{M,h}}(1-z_h) + \sum_{k \in M}w_{h,k}z_{k} \leq 1+\theta < 1.1$.  All random variables $z'_k$ take value in $\set{0, 1/p'}$; moreover, we have $\E[z'_k] = z_k$.   The coefficient before each $z'_k, k \neq h$ in $Z$ is at most $w_{h,k}\leq 1.1$. The coefficient before $z'_h$ is at most $1$ but might as small as $-0.1$. If $z'_h$ is not fixed and the coefficient before it is negative, we define $y = 1/p' - z'_h$ and replace the random variable $z'_h$ with $y$. The Chernoff bound in \Thm{Chernoff} gives that $\calA_h$ happens with probability
$\exp\left(-(\theta' - \theta)^2p'/4\right) = \exp\left(-(\theta' - \theta)^2p/8\right)=e^{-8\log p}=p^{-8}$.

Now focus on ${\calB}_j$ for some $j \in {J_\sfH}$.  Since the unfixed random variable $z'_h$ takes value between $0$ and $1/p' = 2/p$, the Chernoff bound gives that the probability that ${\calB}_j$ occurs is at most $\exp\left(-(\theta' - \theta)^2p/8\right) = p^{-8}$.
In order to apply the uniform LLL,  we need to upper-bound the number of bad events that each ${\calA}_h$ (or ${\calB}_j$) depends on.  
$\cA_h$ and $\cA_k$ are dependent only if $h$ and $k$ are adjacent $G_w$; $\cA_h $ and $\cB_j$ are dependent if there exists $k\in M_j$ such that $h$ and $k$ are adjacent in $G_w$. 
Since each $|\set{k \in M_j:z_k > 0}| \leq p$ and  since the degree of the graph $G_w$ is at most $2qp + 4q \leq 3p^2$ by \Clm{degree-is-small},  any bad event is dependent on at most $(3p^2)p \leq p^8/e$ other events.  Thus the symmetric LLL conditions hold, and thus with positive probability none of the bad events occur.
The polynomial time algorithm follows directly from~\cite{MT10} since the number of bad events is polynomially many.
\end{proof}

\end{document}